\documentclass{vldb}
\usepackage{graphics}
\usepackage{graphicx}
\usepackage{epstopdf}
\usepackage{xspace}
\usepackage{subfigure}
\usepackage{amsmath}
\usepackage{amssymb}
\usepackage{bm}
\usepackage{mathrsfs}
\usepackage{array}
\usepackage{textcomp}
\usepackage{weiwAlgorithm}
\usepackage{url}

\newtheorem{example}{\textbf{Example}}
\newtheorem{remark}{Remark}
\newtheorem{theorem}{\textbf{Theorem}}

\newtheorem{definition}{\textbf{Definition}}

\newcommand{\kc}{$k$-core\xspace}

\newcommand{\krc}{($k$,$r$)-core\xspace}
\newcommand{\krcs}{($k$,$r$)-cores\xspace}
\newcommand{\kkc}{($k$,$k'$)-core\xspace}

\newcommand{\bkrc}{($\mathbf{k}$,$\mathbf{r}$)-\textbf{core}\xspace}

\newcommand{\advenum}{\textbf{AdvEnum}\xspace}
\newcommand{\basicenum}{\textbf{BasicEnum}\xspace}
\newcommand{\advmax}{\textbf{AdvMax}\xspace}
\newcommand{\basicmax}{\textbf{BasicMax}\xspace}
\newcommand{\eeb}{expand branch\xspace}

\begin{document}
\title{When Engagement Meets Similarity: Efficient (k,r)-Core Computation on Social Networks}
\author{{Fan Zhang$^{\dagger}$, Ying Zhang$^{\dagger}$, Lu Qin$^{\dagger}$, Wenjie Zhang$^{\S}$, Xuemin Lin$^{\S}$}
\vspace{3.6mm}
\\
\fontsize{10}{10}
\selectfont\itshape
$^\dagger$University of Technology Sydney, $^\S$University of New South Wales\\
\fontsize{9}{9} \selectfont\ttfamily\upshape
fanzhang.cs@gmail.com, \{ying.zhang, lu.qin\}@uts.edu.au, \{zhangw, lxue\}@cse.unsw.edu.au\\
}

\maketitle
\begin{abstract}
In this paper, we investigate the problem of ($k$,$r$)-core which intends to find cohesive subgraphs on social networks considering both user engagement and similarity perspectives.
In particular, we adopt the popular concept of $k$-core to guarantee the engagement of the users (vertices) in
a group (subgraph) where each vertex in a ($k$,$r$)-core connects to at least $k$ other vertices.
Meanwhile, we also consider the pairwise similarity between users based on their profiles.
For a given similarity metric and a similarity threshold $r$, the similarity between any two vertices in a ($k$,$r$)-core
is ensured not less than $r$. 
Efficient algorithms are proposed to enumerate all \textit{maximal} ($k$,$r$)-cores and
find the \textit{maximum}($k$,$r$)-core, where both problems are shown to be NP-hard.
Effective pruning techniques significantly reduce the search space of two algorithms and
a novel \kkc based \krc size upper bound enhances performance of
the maximum ($k$,$r$)-core computation.
We also devise effective search orders to accommodate the different nature of two mining algorithms.
Comprehensive experiments on real-life data demonstrate that the maximal/maximum ($k$,$r$)-cores enable us to
find interesting cohesive subgraphs, and performance of two mining algorithms is significantly
improved by proposed techniques.
\end{abstract}
\section{Introduction}
\label{sec:introduction}

In social networks, where vertices represent users and edges represent friendship,
there has been a surge of interest to study user engagement in recent years (e.g.,~\cite{DBLP:journals/siamdm/BhawalkarKLRS15,DBLP:journals/iandc/ChitnisFG16,DBLP:conf/aaai/ChitnisFG13,DBLP:conf/cikm/MalliarosV13,DBLP:conf/wsdm/WuSFLT13}).
User engagement models the behavior that each user may engage in, or leave, a community.
In a basic model of \kc, considering positive influence from friends, a user remain engaged if at least $k$ of his/her friends are engaged. This implies that all members are likely to remain engaged if the group is a \kc.
On the other hand, the vertex attribute, along with the topological structure of the network, is associated with specific properties. For instance, a set of keywords is associated with a user to describe the user's research interests in co-author networks and the geo-locations of users are recorded in geo-social networks.
The pairwise similarity has been extensively used to identify a group of similar users based on their attributes (e.g.,~\cite{jaho2011iscode,rangapuram2013towards}).

We study the problem of mining cohesive subgraphs, namely \krcs, which capture vertices with cohesiveness on both structure and similarity perspectives. 
Regarding graph structure, we adopt the \kc model~\cite{seidman1983network} where each vertex connects to at least $k$ other vertices in the subgraph (\textit{structure constraint}).
From the similarity perspective, a set of vertices is cohesive if their pairwise similarities are not smaller than a given threshold $r$ (\textit{similarity constraint}).
Thus, given a number $k$ and a similarity threshold $r$, we say a connected subgraph is a \krc if and only if it satisfies both structure and similarity constraints.
Considering a \textit{similarity graph} in which two vertices have an edge if and only if they are similar, a \krc in the similarity graph is a clique in which every vertex pair has an edge.
To avoid redundant results, we aim to enumerate the \textit{maximal} \krcs. A \krc is maximal if none of its supergraphs is a \krc.
Moreover, we are also interested in finding the \textit{maximum} \krc which is the \krc with the largest number of vertices among all \krcs.
Following is a motivating example.

\begin{figure}
\centering
\subfigure[Graph]{
\includegraphics[width=0.48\columnwidth]{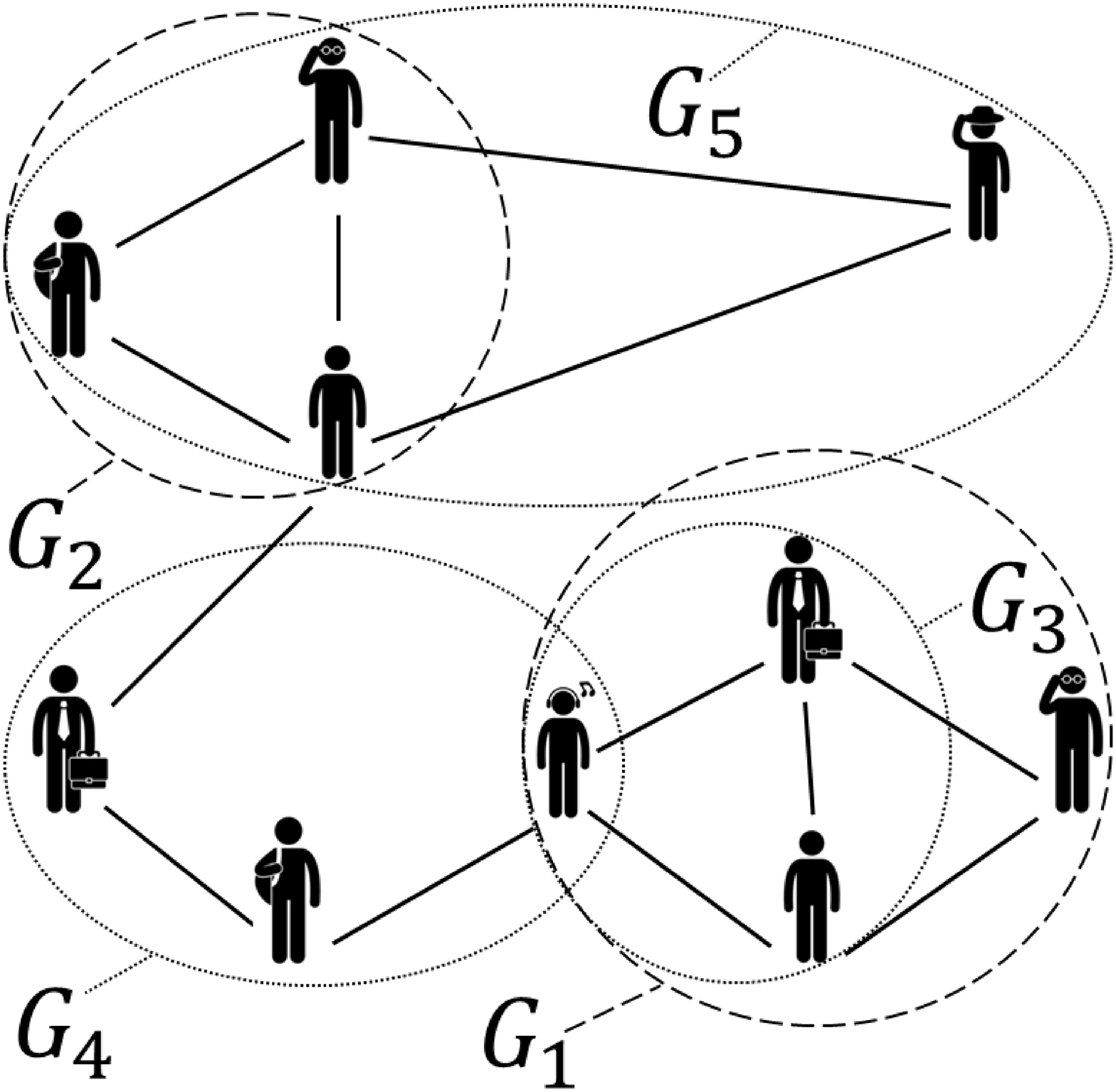}}
\subfigure[Similarity Graph]{
\includegraphics[width=0.48\columnwidth]{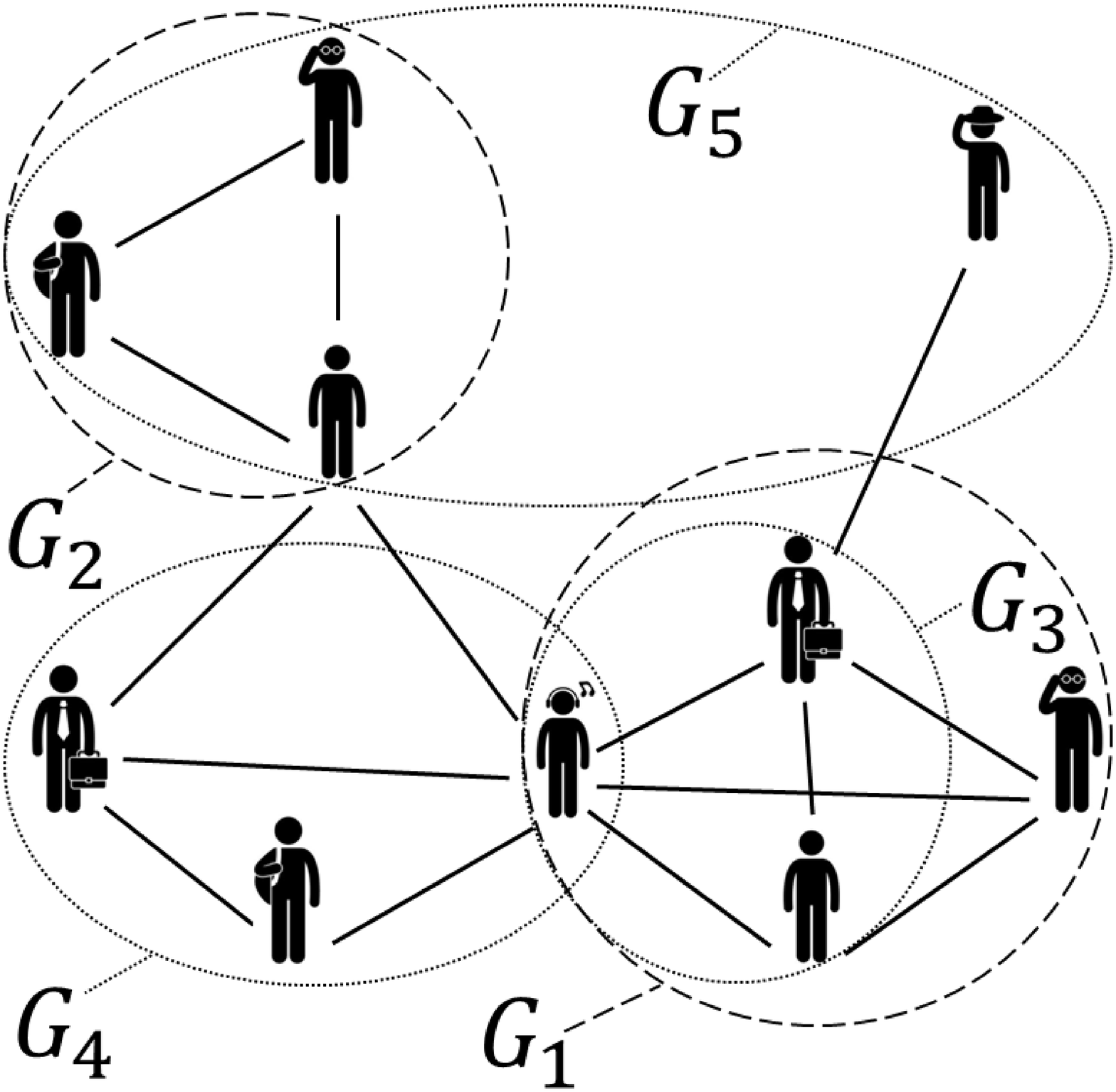}}
\vspace{-5mm}
\caption{Motivating Example}
\label{fig:introduction:motivation}
\end{figure}

\begin{example}[\textbf{Co-author Network}]
Suppose an organization wants to sponsor a number of groups to be continuously involved in a particular activity such as research collaboration.
Two key criteria can be used to evaluate the sustainability of the group: \textit{user engagement} and \textit{user similarity}.
In general, each user in a group should remain engaged and his/her background should be similar to each other.
As illustrated in Figure~\ref{fig:introduction:motivation}(a),
we can model the scholars in DBLP and their co-author relationship as a graph.
In Figure~\ref{fig:introduction:motivation}(b), we construct a corresponding similarity graph in which two scholars have an edge if and only if their research background is similar.
For instance, we may use a set of keywords to describe the background of each user and
the similarity can be measured by Jaccard similarity metric.
Suppose $k=2$ and the similarity threshold $r=0.6$, $G_4$ and $G_5$ are not good candidate groups.
Although each scholar in $G_4$ has similar background with others, their co-author collaboration is weak.
Likewise, although each scholar in $G_5$ co-authors with at least $k$ others, there are some scholars can not collaborate well with others for different research background. However, maximal \krcs (i.e., $G_1$ and $G_2$) can effectively identify good candidate groups because each scholar has at least $k$ co-authors in the same group, and the research similarity for every two scholars is at least $r$.
Note that although $G_3$ is also a \krc, it is less interesting because
it is fully contained by a larger group, $G_1$.
Moreover, knowing the largest possible number of people engaged in the same candidate group
(i.e., the size of the maximum \krc) is also desirable. Using a similar argument, the \krc mining problem studied in this paper can help to identify potential groups for other activities in social networks based on their friendships and personal interests. 
\end{example}

The above examples imply an essential need to enumerate the maximal \krcs, and find the maximum \krc.
Recently, various extensions of \kc have been proposed in the literature (e.g.,~\cite{DBLP:journals/pvldb/FangCLH16,lee2014cast}) but the problems studied are inherently different
to \krc.
For instance, Fang {\em et al.}~\cite{DBLP:journals/pvldb/FangCLH16} aims to find attributed communities for a given query node and a set of query keywords. Specifically, the result subgraph is a \kc and the number of common query keywords is maximized for all vertices in the subgraph.


\vspace{1mm}
\noindent \textbf{Challenges and Contributions.}
Although a linear algorithm for \kc computation~\cite{DBLP:journals/corr/cs-DS-0310049} (i.e., only consider structure constraint) is readily available,
we show that the problem of enumerating all maximal \krcs and finding the maximum \krc are both NP-hard because of
the similarity constraint involved.
A straightforward solutions is the simple combination of the existing \kc and clique algorithms,
e.g., enumerating the cliques on similarity graph (e.g., the similarity graph in Figure~\ref{fig:introduction:motivation}(b)) and then checking the structure constraint on the graph (e.g., graph in Figure~\ref{fig:introduction:motivation}(a)).
In Section~\ref{sec:disc} and the empirical study, we show that this is not promising because of the
isolated processing of structure and similarity constraints.
In this paper, we show that the performance can be immediately improved by
considering two constraints (i.e., two pruning rules) at the same time without explicitly materializing the similarity graph.
Then our technique contributions focus on further significantly reducing the search space of two mining algorithms
from the following three aspects:
($i$) effective pruning, early termination and maximal check techniques.
($ii$) \kkc based approach to derive tight upper bound for the problem of finding the maximum \krc.
and ($iii$) good search orders in two mining algorithms.
Following is a summary of our principal contributions.


\begin{itemize}
\item We advocate a novel cohesive subgraph model for attributed graphs, called \krc, to capture the cohesiveness of subgraphs from both the graph structure and the vertex attributes. We prove that the problem of enumerating all maximal \krcs and finding the maximum \krc are both NP-hard.\hfill (Section~\ref{sec:preli})
\item We develop efficient algorithms to enumerate the maximal \krcs with candidate pruning, early termination and maximal checking techniques. \hfill (Section~\ref{sec:enumerate})
\item We also develop an efficient algorithm to find the maximum \krc. Particularly, a novel \kkc based approach is proposed to derive a tight upper bound for the size of the candidate solution. \hfill (Section~\ref{sec:max})
\item Based on some key observations, we propose three search orders for enumerating maximal \krcs, checking maximal \krcs, and finding maximum \krc algorithms, respectively. \hfill ~(Section~\ref{sec:order})
\item Our empirical studies on real-life data demonstrate that interesting cohesive subgraphs can be identified by maximal \krcs and maximum \krc. The extensive performance evaluation shows that the techniques proposed in this paper can significantly improve the performance of two mining algorithms. \hfill (Section~\ref{sec:exp})
\end{itemize}


%
%

\section{Preliminaries}
\label{sec:preli}
In this section, we first formally introduce the concept of \krc,
then show that the two problems are NP-hard.
Table~\ref{tb:notations} summarizes the mathematical notations used throughout this paper.
\begin{table}
\small
  \centering
    \begin{tabular}{|c|l|}
      \hline
      \textbf{Notation}   & \textbf{Definition}                   \\ \hline \hline

      $G$    &  a simple attributed graph  \\ \hline

      $S$,$J$,$R$  &  induced subgraph or corresponding vertices  \\ \hline

      $u$, $v$ &  vertices in the attributed graph \\ \hline

      $sim(u,v)$  & similarity between $u$ and $v$  \\ \hline

      $deg(u,S)$  & number of adjacent vertex of $u$ in $S$ \\ \hline

	  $deg_{min}(S)$ & minimal degree of the vertices in $S$ \\ \hline

      $DP(u,S)$  & number of dissimilar vertices of $u$ w.r.t $S$ \\ \hline
		
	  $DP(S)$  & number of dissimilar pairs of $S$ \\ \hline

      $SP(u,S)$  & number of similar vertices of $u$ w.r.t $S$ \\ \hline

      $M$  &  vertices chosen so far in the search  \\ \hline

      $C$  &  candidate vertices set in the search \\ \hline

      $E$  &  relevant exclusive vertices set in the search \\ \hline

	  $\mathcal{R}$($M$,$C$) & maximal \krcs derived from $M \cup C$ \\ \hline

\end{tabular}
\vspace{-2mm}
\caption{The summary of notations}
\label{tb:notations}
\end{table}
\subsection{Problem Definition}
\label{subsec:prob}
We consider an undirected, unweighted, and attributed graph $G=(V, \mathcal{E},A)$, where $V(G)$ (resp. $\mathcal{E}(G)$) represents the set of vertices (resp. edges) in $G$,
and $A(G)$ denotes the attributes of the vertices.
By $sim(u,v)$, we denote the similarity of two vertices $u$, $v$ in $V(G)$ which is derived from their corresponding attribute values (e.g., users' geo-locations and interests) such as Jaccard similarity or Euclidean distance.
For a given similarity threshold $r$, we say two vertices are dissimilar (resp. similar)
if $\mathop{sim}(u,v) < r$ (resp. $\mathop{sim}(u,v) \geq r$)~\footnote{Following the convention, when the distance metric (e.g., Euclidean distance) is employed, we say two vertices are similar if their distance is not larger than the given distance threshold.}.

For a vertex $u$ and a set $S$ of vertices, $DP(u,S)$ (resp. $SP(u,S)$) denotes the number of other vertices in $S$ which are dissimilar (resp. similar) to $u$ regarding the given similarity threshold $r$.
We use $DP(S)$ denote the number of dissimilar pairs in $S$.
We use $S \subseteq G$ to denote that $S$ is a subgraph of $G$ where
$\mathcal{E}(S) \subseteq \mathcal{E}(G)$ and $A(S) \subseteq A(G)$.
By $deg(u,S)$, we denote the number of adjacent vertices of $u$ in $V(S)$.
Then, $deg_{min}(S)$ is the minimal degree of the vertices in $V(S)$.
Now we formally introduce two constraints.
\begin{definition} \textbf{Structure Constraint.}
Given a positive integer $k$, a subgraph $S$ satisfies the structure constraint if
$deg(u,S) \geq k$ for each vertex $u \in V(S)$, i.e., $deg_{min}(S) \geq k$.
\end{definition}
\begin{definition} \textbf{Similarity Constraint.}
Given a similarity threshold $r$, a subgraph $S$ satisfies the similarity constraint if
$DP(u,S) = 0$ for each vertex $u \in V(S)$, i.e., $DP(S)=0$.
\end{definition}

We then formally define the \krc based on structure and similarity constraints.
\begin{definition} \bkrc.
Given a connected subgraph $S \subseteq G$, $S$ is a \krc if $S$ satisfies both structure and similarity constraints.
\end{definition}

In this paper, we aim to find all maximal \krcs and the maximum \krc, which are defined as follows. 

\begin{definition} \textbf{Maximal} \bkrc.
Given a connected subgraph $S \subseteq G$, $S$ is a maximal \krc if $S$ is a \krc of $G$ and there exists no \krc $S'$ of $G$ such that $S \subset S'$.
\end{definition}

\begin{definition} \textbf{Maximum} \bkrc.
Let $\mathcal{R}$ denote all \krcs of an attributed graph $G$, a \krc $S \subseteq G$ is maximum if $|V(S)| \geq |V(S')|$ for every \krc $S' \in \mathcal{R}$.
\end{definition}

\noindent \textbf{Problem Statement}.
Given an attributed graph $G$, a positive integer $k$ and a similarity threshold $r$, we aim to develop efficient algorithms for the following two fundamental problems:
($i$) enumerating all maximal \krcs in $G$;
($ii$) finding the maximum \krc in $G$.

\begin{example}
In Figure~\ref{fig:introduction:motivation}, all vertices are from the \kc where $k=2$.
$G_1$, $G_2$ and $G_3$ are the three \krcs.
$G_1$ and $G_2$ are maximal \krcs while $G_3$ is fully contained by $G_1$.
$G_1$ is the maximum \krc.
\end{example}

\subsection{Problem Complexity}
\label{subsec:complexity}

We can compute \kc in linear time by recursively removing the vertices with a degree of less than $k$~\cite{DBLP:journals/corr/cs-DS-0310049}.
Nevertheless, the two problems studied in this paper are NP-hard due to the additional similarity constraint.

\begin{theorem}
\label{the:hardnesslemma}
Given a graph $G(V,\mathcal{E})$, the problems of enumerating all maximal \krcs and finding the maximum \krc are NP-hard.
\end{theorem}
\begin{proof}
Given a graph $G(V,\mathcal{E})$, we construct an attributed graph $G'(V',\mathcal{E}',A')$ as follows.
Let $V(G') = V(G)$ and $\mathcal{E}(G')=\{(u,v)~|~u\in V(G'), v\in V(G'), u\neq v\}$, i.e., $G'$ is a complete graph.
For each $u\in V(G')$, we let $A(u) = adj(u,G)$ where $adj(u,G)$ is the set of adjacent vertices of $u$ in $G$.
Suppose a Jaccard similarity is employed, i.e., $\mathop{sim}(u,v) = \frac{|A(u)\cap A(v)|}{|A(u) \cup A(v)|}$ for
any pair of vertices $u$ and $v$ in $V(G')$, and let the similarity threshold $r = \epsilon$ where $\epsilon$ is an infinite small positive number (e.g., $\epsilon=\frac{1}{2|V(G')|}$).
We have $\mathop{sim}(u,v) \geq r$ if the edge $(u,v) \in \mathcal{E}(G)$, and otherwise $\mathop{sim}(u,v)=0 < r$.
Since $G'$ is a complete graph, i.e., every subgraph $S \subseteq G'$ with $|S| \geq k$ satisfies the structure constraint of a \krc,
the problem of deciding whether there is a $k$-clique on $G$ can be reduced to the problem of finding a \krc on $G'$ with $r = \epsilon$, and hence can be solved by the problem of enumerating all maximal \krcs or finding the maximum \krc. Theorem~\ref{the:hardnesslemma} holds due to the NP-hardness of the $k$-clique problem~\cite{DBLP:books/fm/GareyJ79}.
\end{proof}

\section{The Clique-based Method}
\label{sec:disc}

Let $G'$ denote a new graph named \textit{similarity graph} with $V(G') = V(G)$ and $\mathcal{E}(G') = \{(u,v)~|~sim(u, v) \geq r~\&~u,v\in V(G)\}$, i.e., $G'$ connects the similar vertices in $V(G)$.
According to the proof of Theorem~\ref{the:hardnesslemma},
it is clear that the set of vertices in a \krc satisfies the structure constraint on $G$ and is a \textit{clique} in the similarity graph $G'$.
This implies that we can use the existing clique algorithms on the similarity graph to enumerate the \krc candidates, followed by a structure constraint check.
More specifically, we may first construct the similarity graph $G'$ by computing the pairwise similarity of the vertices. Then we enumerate the cliques in $G'$, and compute the \kc on each induced subgraph of
$G$ for each clique. We uncover the maximal \krcs after the maximal check.
We may further improve the performance of this clique-based method in the following three ways.
\begin{itemize}
\item  Instead of enumerating cliques on the similarity graph $G'$, we can first compute the \kc of $G$, denoted by $S$.
Then, we apply the clique-based method on the similarity graph of each connected subgraph in $S$.

\item  An edge in $S$ can be deleted if its corresponding vertices are \textit{dissimilar}, i.e., there is no edge between these two vertices in similarity graph $S'$.

\item  We only need to compute \kc for each maximal clique because any maximal \krc derived from a non-maximal clique can be obtained from the maximal cliques.
\end{itemize}

Although the above three methods substantially improve the performance of the clique-based method,
our empirical study demonstrates that the improved clique-based method is still significantly outperformed by our baseline algorithm (Section~\ref{subsec:exp efficiency}).
This further validates the efficiency of the new techniques proposed in this paper.

\section{Warming Up for Our Approach}
\label{sec:naive}
In this section, we first present a naive solution with a proof of the algorithm's correctness in Section~\ref{subsec:naive}.
Then, Section~\ref{subsec:wm_analysis} shows the limits of the naive approach
and briefly introduces the motivations for our proposed techniques.

\subsection{Naive Solution}
\label{subsec:naive}
For ease of understanding, we start with a straightforward set enumeration approach . The pseudo-code is given in  Algorithm~\ref{alg:naive}.
At the initial stage (Line~\ref{alg:naive_remove}-\ref{alg:naive_remove end}), we remove the edges in $\mathcal{E}(G)$
whose corresponding vertices are dissimilar, and then
compute the $k$-core $\mathcal{S}$ of the graph $G$.
For each connected subgraph $S \in \mathcal{S}$,
the procedure NaiveEnum (Line~\ref{alg:naive_enum_1}) identifies all possible \krcs
by enumerating and validating all the induced subgraphs of $S$.
By $\mathcal{R}$, we record the \krcs seen so far. Lines~\ref{alg:naive_check_s}-\ref{alg:naive_check_e} eliminate the non-maximal $(k,r)$-cores by checking all ($k$,$r$)-cores.

\begin{algorithm}[htb]
\SetVline 
\SetFuncSty{textsf}
\SetArgSty{textsf}
\small
\caption{\bf EnumerateMKRC($G$, $k$, $r$)}
\label{alg:naive}
\Input
{
   $G:$ attributed graph, $k:$ degree threshold,
   $r:$ similarity threshold
}
\Output{$\mathcal{M}:$ Maximal \krcs}
\For{each edge $(u,v)$ in $\mathcal{E}(G)$}
{
    \label{alg:naive_remove}
	\State{Remove edge $(u,v)$ from $G$ \textbf{If} $sim(u,v)<r$}
    \label{alg:naive_remove end}
}
\State{$\mathcal{S} \leftarrow$ \textbf{k-core}($G$); $\mathcal{R}:=\emptyset$} \label{alg:naive_getS}
\For{each connected subgraph $S$ in $\mathcal{S}$}
{
	\label{alg:naive_s}
	\State{$\mathcal{R} := \mathcal{R}~\cup$ \textbf{NaiveEnum}($\emptyset$, $S$)}
	\label{alg:naive_enum_1}
}
\For{each \krc $R$ in $\mathcal{R}$}
{
	\label{alg:naive_check_s}
    \If {there is a \krc $R' \in \mathcal{R}$ s.t. $R \subset R'$}
	{
		\State{$\mathcal{R}:= \mathcal{R} \setminus R$}
		\label{alg:naive_check_e}
	}
}
\Return{$\mathcal{R}$}
\end{algorithm}

\begin{algorithm}[htb]
\SetVline 
\SetFuncSty{textsf}
\SetArgSty{textsf}
\small
\caption{\bf NaiveEnum($M$, $C$)}
\label{alg:naive_enum}
\Input
{
   $M :$ chosen vertices, $C :$ candidate vertices
}
\Output{$\mathcal{R}:$ \krcs}
\If{$C = \emptyset$ and $deg_{min}(M)\geq k$ and  $DP(M)=0 $ }
{
    \label{alg:nenum validate}
	\State{$\mathcal{R}:= \mathcal{R}$ $\cup$ $R$ \textbf{for} every connected subgraph $R\in M$}
}
\Else
{
	\State{$u \leftarrow$ choose a vertex in $C$}
    \label{alg:nenum select u}
	\StateCmt{\textbf{NaiveEnum}($M \cup u$, $C \setminus u$)}{\textbf{Expand}}
    \label{alg:nenum expand}
	\StateCmt{\textbf{NaiveEnum}($M$, $C \setminus u$ )}{\textbf{Shrink}}
    \label{alg:nenum shrink}
}
\end{algorithm}

During the NaiveEnum search procedure (Algorithm~\ref{alg:naive_enum}),
the vertex set $M$ incrementally retains the chosen vertices, and $C$ retains the candidate vertices.
As shown in Figure~\ref{fig:naive:searchtree}, the enumeration process corresponds to a \textit{binary search tree} in which each leaf node represents a subset of $S$.
In each non-leaf node, there are two branches. The chosen vertex will be moved to $M$ from $C$ in \textit{expand branch}, and will be deleted from $C$ in \textit{shrink branch}, respectively.

\begin{figure}[htb]
\centering
\includegraphics[width=1\columnwidth]{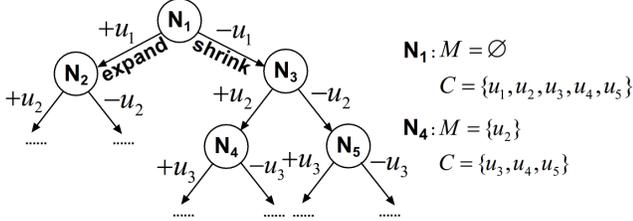}
\vspace{-8mm}
\caption{Example of the Search Tree}
\label{fig:naive:searchtree}
\end{figure}

\vspace{1mm}
\noindent \textbf{Algorithm Correctness.}
We can safely remove dissimilar edges (i.e., edges whose corresponding vertices are dissimilar) at Line~\ref{alg:naive_remove} and~\ref{alg:naive_remove end}, since they will not be considered in any \krc due to the similarity constraint.
For every \krc $R$ in $G$, there is one and only one connected \kc subgraph $S$ from $\mathcal{S}$ with $R \subseteq S$.
Since all possible subsets of $S$ (i.e., $2^{|S|}$ leaf nodes) are enumerated in the corresponding search tree, every \krc $R$ can be accessed \textit{exactly once} during the search.
Together with the structure/similarity constraints and maximal property validation,
we can output all maximal \krcs.

Algorithm~\ref{alg:naive} immediately finds the maximum \krc by returning the maximal \krc with the largest size.

\subsection{Limits of the Naive Approach}
\label{subsec:wm_analysis}
It is cost-prohibitive to enumerate all subsets of each $k$-core subgraph $S$ in Algorithm~\ref{alg:naive}.
Below, we briefly discuss the limits of the above algorithm and the ways to address these issues.

\vspace{0.5mm}
\noindent \textbf{Large Candidate Size.}
We aim to reduce the candidate size by explicitly/implicitly excluding some vertices in $C$.

\vspace{0.5mm}
\noindent \textbf{No Early Termination Check.}
We carefully maintain some excluded vertices to foresee whether all possible \krcs in the subtree are not maximal, such that the search may be terminated early.

\vspace{0.5mm}
\noindent \textbf{Unscalable Maximal Check.}
We devise a new approach to check the maximal property by
extending the current solution $R$ by the excluded vertices.

\vspace{0.5mm}
\noindent \textbf{No Upper Bound Pruning.}
To find the maximum \krc, we should use the largest size of the \krcs seen so far
to terminate the search in some non-promising subtrees, so it is critical to estimate the upper bound size of \krcs in the candidate solution. We devise a $(k,k')$-core based upper bound which is much tighter than the upper bound derived from state-of-the-art techniques.

\vspace{0.5mm}
\noindent \textbf{Search Order Not Considered.}
In Algorithm~\ref{alg:naive}, we do not consider the visiting order of the vertices chosen from $C$
along with the order of two branches.
Our empirical study shows that an improper search order may result in very poor performance,
even if all of other techniques are employed.
Considering the different nature of the problems, we should devise effective search orders
for enumeration, finding the maximum, and checking maximal algorithms.

\section{Finding All Maximal (k,r)-Cores}
\label{sec:enumerate}

In this section, we propose pruning techniques for the enumeration algorithm including candidate reduction, early termination, and maximal check techniques, respectively.
Note that we defer discussion on search orders to Section~\ref{sec:order}.

\subsection{Reducing Candidate Size}
\label{subsec:enum_prune}
We present pruning techniques to explicitly/implicitly exclude some vertices from $C$.

\subsubsection{Eliminating Candidates}
\label{subsubsec:elim cand}
Intuitively, when a vertex in $C$ is assigned to (i.e., expand branch) $M$ or discarded (i.e., shrink branch), we shall recursively remove some non-promising vertices from $C$ due to structure and
similarity constraints.
The following two pruning rules are based on the definition of \krc.
\begin{theorem}
\label{the:str_prune}
\textbf{Structure based Pruning}. We can discard a vertex $u$ in $C$ if $deg(u, M \cup C)~<~k$.
\end{theorem}
\begin{theorem}
\label{the:sim_prune}
\textbf{Similarity based Pruning}. We can discard a vertex $u$ in $C$ if $\mathop{DP}(u, M)>0$.
\end{theorem}
\vspace{0.5mm}

\noindent \textbf{Candidate Pruning Algorithm.} If a chosen vertex $u$ is extended to $M$ (i.e., to the expand branch),
we first apply the similarity pruning rule (Theorem~\ref{the:sim_prune}) to exclude vertices in $C$
which are dissimilar to $u$.
Otherwise, none of the vertices will be discarded by the similarity constraint when we follow the shrink branch.
Due to the removal of the vertices from $C$ (expand branch) or $u$ (shrink branch),
we conduct structure based pruning by computing the $k$-core for vertices in $M \cup C$.
Note that the search terminates if any vertex in $M$ is discarded.

It takes at most $O(|C|)$ time to find dissimilar vertices of $u$ from $C$.
Due to the \kc computation, the structure based pruning takes linear time to the number of
edges in the induced graph of $M \cup C$.



\begin{figure*}
	\begin{minipage}[b]{\linewidth}
		\centering
		\includegraphics[width=0.22\linewidth]{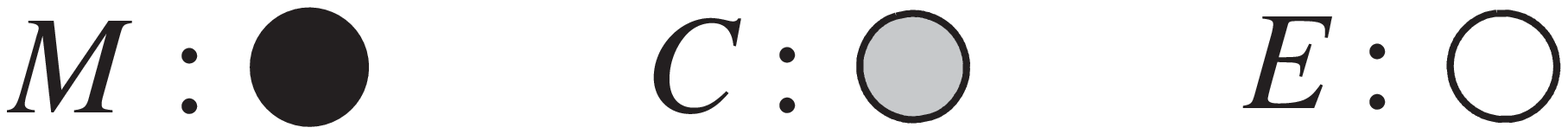}%
        \vspace{-4mm}
	\end{minipage}

  \subfigure[Pruning Candidates]{
          \includegraphics[width=0.24\linewidth]{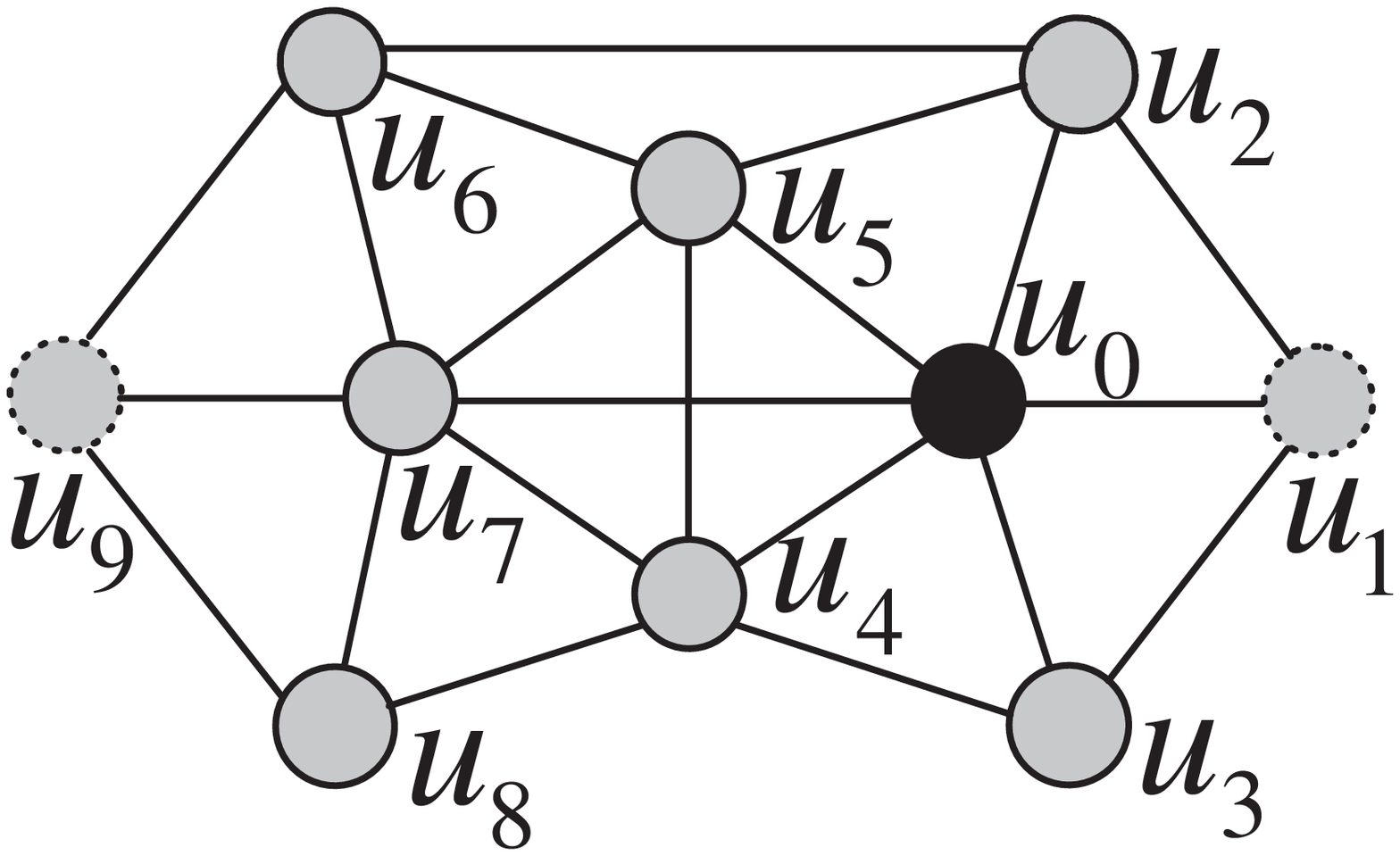}}
  \subfigure[Retaining Candidates]{
          \includegraphics[width=0.24\linewidth]{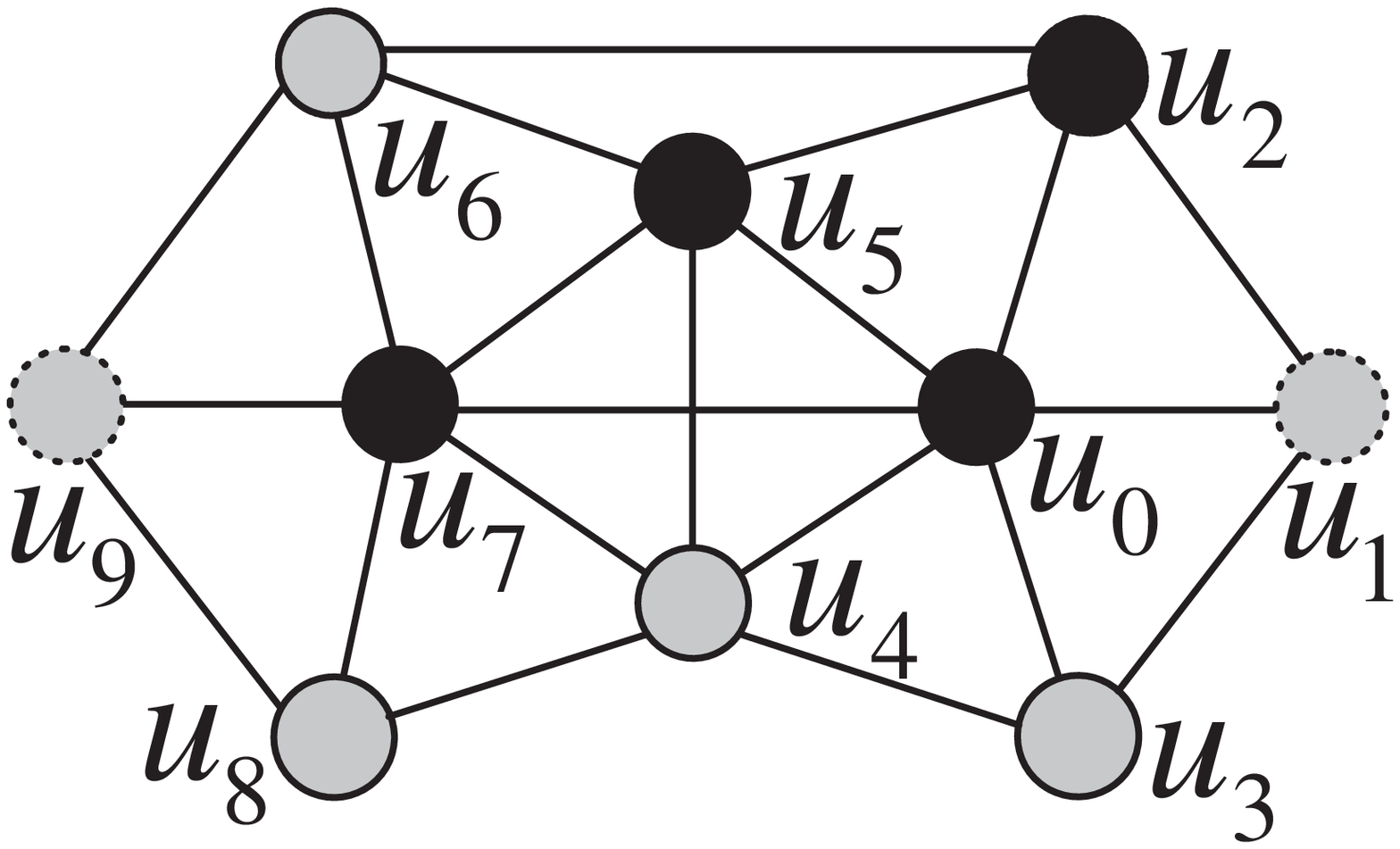}}
  \subfigure[Early Termination]{
          \includegraphics[width=0.24\linewidth]{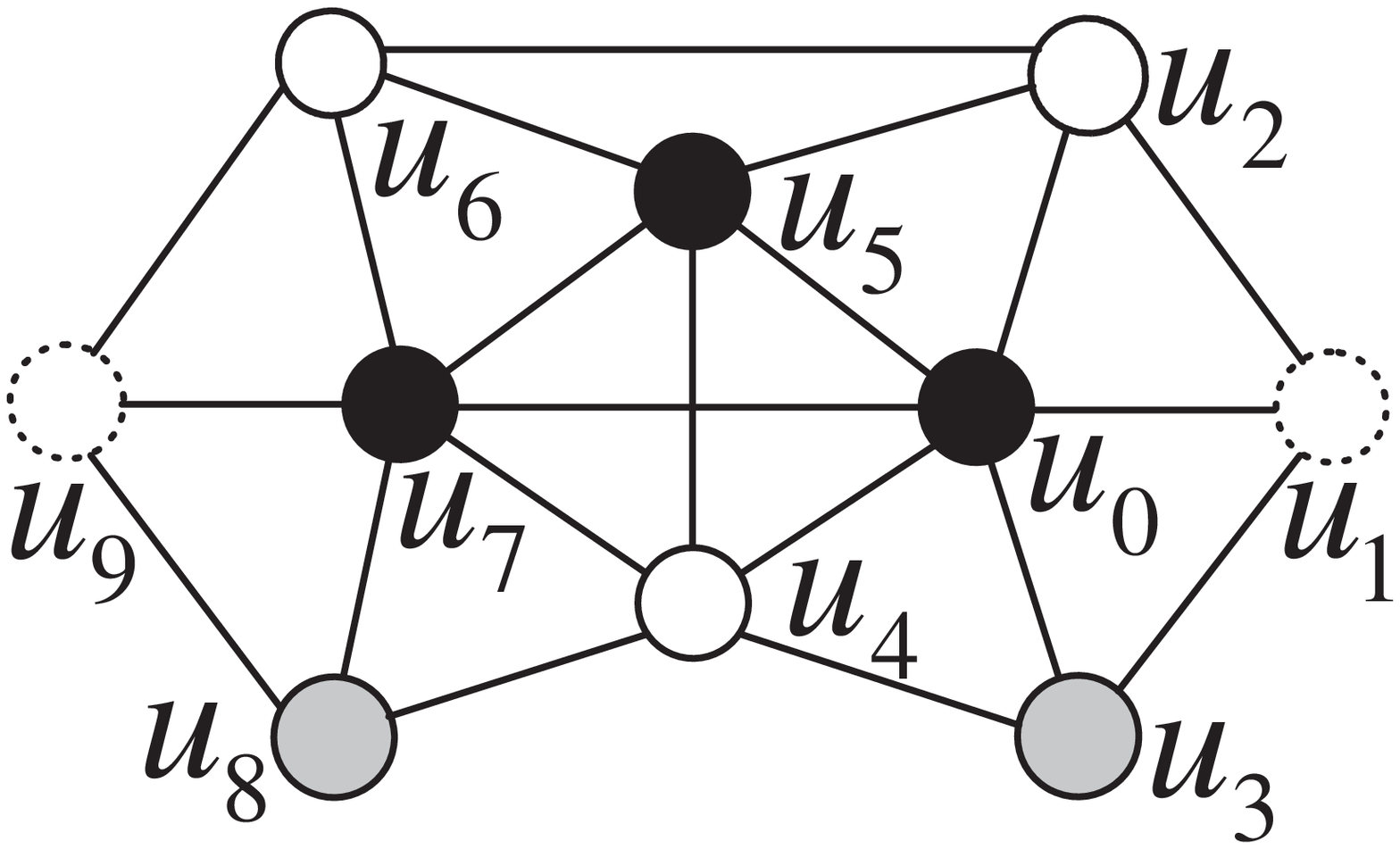}}
  \subfigure[Checking Maximals]{
          \includegraphics[width=0.24\linewidth]{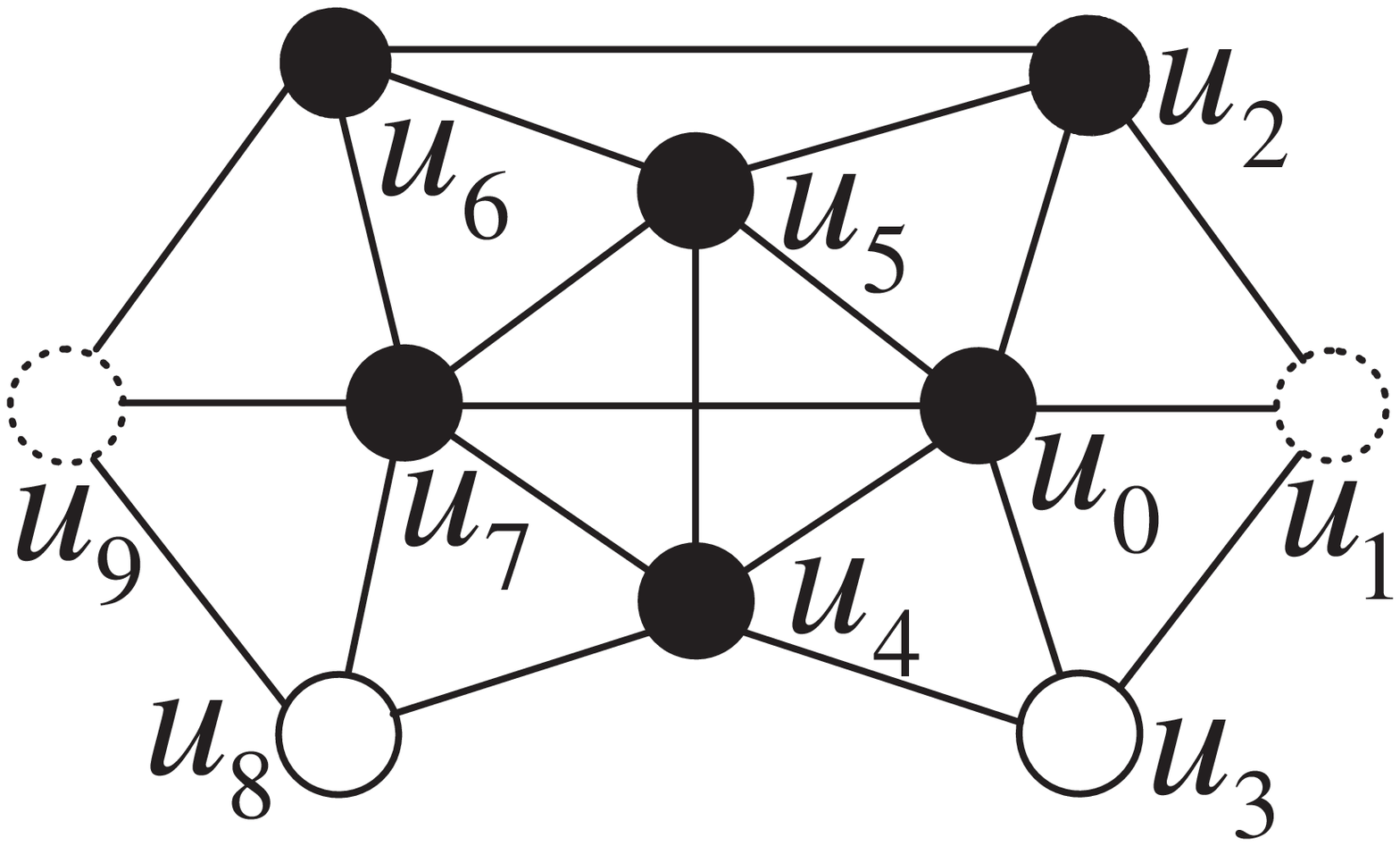}}

\vspace{-3mm}
\caption{Examples for Techniques}
\label{fig:enumerate:examples}
\vspace{-3mm}
\end{figure*}

\begin{example}
In Figure~\ref{fig:enumerate:examples} (a), we have $M = \{u_0 \}$
and $C=$ $\{u_1, \ldots, u_9 \}$. In our running examples, $u_1$ and $u_9$ is the only dissimilar pair.
Suppose $u_1$ is chosen from $C$,
following the expand branch, $u_1$ will be extended to $M$ and then $u_9$ will be pruned due to the similarity constraint.
Then we need remove $u_8$ as $deg(u_8, M \cup C)<3$. Thus, $M = \{u_0, u_1\}, C = \{u_2, \ldots, u_7\}$.
Regarding the shrink branch, $u_1$ is explicitly discarded,
which in turn leads to the deletion of $u_3$ due to structure constraint. Thus, $M = \{u_0\}, C = \{u_2, u_4, \ldots, u_9\}$.
\end{example}

After applying the candidate pruning, following two important invariants always hold at each search node.\vspace{2mm}

\noindent \textbf{Similarity Invariant.} We have
\begin{equation}
\label{eqn:sim_invariant}
\mathop{DP}(u, M \cup C)=0 \text{~~for every vertex~} u\in M
\end{equation}
That is, $M$ satisfies similarity constraint regarding $M \cup C$.\vspace{1mm}\\
\noindent \textbf{Degree Invariant.} We have
\begin{equation}
\label{eqn:str_invariant}
deg_{min}(M \cup C) \geq k
\end{equation}
That is, $M$ and $C$ together satisfy the structure constraint.

\subsubsection{Retaining Candidates}
\label{subsubsec:retain cand}
In addition to explicitly pruning some non-promising vertices,
we may implicitly reduce the candidate size by not choosing some vertices from $C$.
In this paper, we say a vertex $u$ is \textit{similarity free} w.r.t $C$
if $u$ is similar to all vertices in $C$, i.e., $DP(u, C)=0$.
By $SF(C)$ we denote the set of similarity free vertices in $C$.


\begin{theorem}
\label{the:retain}
Given that the pruning techniques are applied in each search step, we do not need to choose vertices from $SF(C)$ on both expand and shrink branches.
Moreover, $M \cup C$ is a \krc if we have $C=SF(C)$.
\end{theorem}
\begin{proof}
For every vertex $u \in SF(C)$, we have $DP(u, M \cup C)$ $= 0$
due to the similarity invariant of $M$ (Equation~\ref{eqn:sim_invariant}) and the definition of $SF(C)$.
Let $M_1$ and $C_1$  denote the corresponding chosen set and candidate after
$u$ is chosen for expansion. Similarly, we have $M_2$ and $C_2$ if
$u$ is moved to the shrink branch.
We have $M_2 \subset M_1$ and $C_2 \subseteq C_1$,
because there are no discarded vertices when $u$ is extended to $M$
while some vertices may be eliminated due to the removal of $u$ in the shrink branch.
This implies that $\mathcal{R}(M_2,C_2) \subseteq \mathcal{R}(M_1,C_1)$.
Consequently, we do not need to explicitly discard $u$ as the shrink branch of $u$ is useless.
Hence, we can simply retain $u$ in $C$ in the following computation.

However, $C=SF(C)$ implies every vertex $u$ in $M \cup C$
satisfies the similarity constraint.
Moreover, $u$ also satisfies the structure constraint due to the degree invariant (Equation~\ref{eqn:str_invariant}) of $M \cup C$. Consequently, $M \cup C$ is a \krc.
\end{proof}

Note that a vertex $u \in SF(C)$ may be discarded in the following search due to the strutural constraint.
Otherwise, it is moved to $M$ when the condition $SF(C)=C$ holds.
For each vertex $u$ in $C$, we can update $DP(u,C)$ in a passive way when
its dissimilar vertices are eliminated from the computation.
Thus it takes $O(n_d)$ time in the worst case where $n_d$ denote the number of dissimilar pairs in $C$.

\begin{remark}
With similar rationale, we can move a vertex $u$ directly from $C$ to $M$
if it is similarity free (i.e., $u \in SF(C)$) and is adjacent to at least $k$ vertices in $M$.
As this validation rule is trivial, it will be used in this paper without further mention.
\end{remark}

\begin{example}
In Figure~\ref{fig:enumerate:examples} (b), suppose we have $M = \{u_0, u_2, u_5, u_7\}$, $C =$ $\{u_1, u_3,u_4, u_6, u_8, u_9 \}$. $u_1$ and $u_9$ is the only dissimilar pair.
We have $SF(C)=$ $\{u_3, u_4,$ $u_6, u_8 \}$ in which the vertices will not be chosen by Theorem~\ref{the:retain}. If we choose $u_1$ in the expand branch, the search will be terminated by $SF(C) = C = \{u_3, u_4, u_6\}$; if in the shrink branch, the search will be terminated by $SF(C) = C = \{u_4, u_6, u_8, u_9\}$.
Note that we may directly move $\{u_4, u_6\}$ to $M$ since $deg(u_4, M) \geq 3$ and $deg(u_6, M) \geq 3$.
\end{example}

\subsection{Early Termination}
\label{subsec:enum_term}

\noindent \textit{Trivial Early Termination.} There are two trivial early termination rules.
As discussed in Section~\ref{subsec:enum_prune},
we immediately terminate the search if any vertex in $M$ is discarded due to the structure constraint.
We also terminate the search if $M$ is disconnected to $C$.
Both these stipulations will be applied in the remainder of this paper without further mention.

In addition to identifying the subtree that cannot derive any \krc,
we further reduce the search space by identifying the subtrees that cannot lead to any \textit{maximal} \krc.
By $E$, we denote the related excluded vertices set for a search node of the tree,
where the discarded vertices during the search are retained if they are similar to $M$,
i.e., $DP(v, M)=0$ for every $v \in E$ and $E~\cap$ $(M \cup C)$ $=\emptyset$.
We use $SF_{C}(E)$ to denote the similarity free vertices in $E$ w.r.t the set $C$;
that is, $DP(u,C)$ $= 0$ for every $u \in SF_{C}(E)$.
Similarly, by $SF_{C \cup E}(E)$ we denote the similarity free vertices in $E$ w.r.t
the set $E \cup C$.

\begin{theorem}
\label{the:early_terminate}
\textbf{Early Termination.} We can safely terminate the current search if one of the following two conditions hold:\\
($i$) there is a vertex $u \in SF_{C}(E)$ with $deg(u,M) \geq k$;\\
($ii$) there is a set $ U \subseteq SF_{C\cup E}(E)$,
such that $deg(u, M \cup U) \geq k$ for every vertex $u \in U$.
\end{theorem}
\begin{proof}
($i$) We show that every \krc $R$ derived from current $M$ and $C$ (i.e., $R\subseteq \mathcal{R}(M, C)$)
can reveal a larger \krc by attaching the vertex $u$.
For any $R \in \mathcal{R}$, we have $deg(u, R) \geq k$ because $deg(u, M) \geq k$ and $M \subseteq V(R)$.
$u$ also satisfies the similarity constraint based on the facts that $u \in SF_{C}(E)$ and
$R \subseteq M \cup C$. Consequently, $V(R) \cup \{u\}$ is a \krc.
($ii$) The correctness of condition ($ii$) has a similar rationale.
The key idea is that for every $u \in U$, $u$ satisfies the structure constraint because $deg(u, M \cup U) \geq k$;
and $u$ also satisfies the similarity constraint because $U \subseteq$ $SF_{E\cup C}(E)$
implies that $DP(u, U \cup R)$ $=0$.
\end{proof}

\vspace{0.5mm}
\noindent \textbf{Early Termination Check.} It takes $O(|E|)$ time to check the condition ($i$) of Theorem~\ref{the:early_terminate} with one scan of the vertices in $SF_{C}(E)$
Regarding condition ($ii$), we may conduct $k$-core computation on $M \cup SF_{C\cup E}(E)$
to see if a subset of $SF_{C\cup E}(E)$ is included in the $k$-core.
The time complexity is $O(n_e)$ where $n_e$ is the number of edges in the induced graph of $M \cup C \cup E$.

\begin{example}
In Figure~\ref{fig:enumerate:examples} (c), suppose we have $M = \{u_0, u_5, u_7\}$,
$C = \{u_3, u_8\}$ and $E = \{u_1, u_2, u_4, u_6, u_9\}$.
Recall that $u_1$ and $u_9$ is the only dissimilar pair.
Therefore, we have $SF_{C}(E)=\{u_1, u_2, u_4, u_6, u_9\}$
and $SF_{C\cup E}(E)=\{u_2, u_4, u_6\}$.
According to Theorem~\ref{the:early_terminate} ($i$), the search is terminated because there is a vertex $u_4 \in SF_{C}(E)$ with $deg(u_4,M) \geq 3$.
We may also terminate the search according to Theorem~\ref{the:early_terminate} ($ii$)
because we have $U = \{u_2, u_6\} \subset SF_{C\cup E}(E)$ such that $deg(u_2,M \cup U) \geq 3$ and $deg(u_6,M \cup U) \geq 3$.
\end{example}

\subsection{Checking Maximal}
\label{subsec:enum_check}
In Algorithm~\ref{alg:naive} (Lines~\ref{alg:naive_check_s}-\ref{alg:naive_check_e}),
we need to validate the maximal property based on all \krcs of $G$.
The cost significantly increases with both the number and the average size of the \krcs.
Similar to the early termination technique, we use the following rule to check the maximal property.

\begin{theorem}
\label{the:maximal_check}
\textbf{Checking Maximals.}
Given a \krc $R$, $R$ is a maximal \krc if there doesn't exist a non-empty set $U \subseteq E$
such that $R \cup U$ is a \krc, where $E$ is the excluded vertices set
when $R$ is generated.
\end{theorem}
\begin{proof}
$E$ contains all discarded vertices that are similar to $M$ according to the definition of the excluded vertices set.
For any \krc $R'$ which fully contains $R$, we have $R' \subseteq E \cup R$
because $R = M$ and $C = \emptyset$,
i.e., the vertices outside of $E \cup R$ cannot contribute to $R'$.
Therefore, we can safely claim that $R$ is maximal if we cannot find $R'$ among $E \cup R$.
\end{proof}

\begin{example}
In Figure~\ref{fig:enumerate:examples} (d), we have $C$ $=\emptyset$,
$M$ $= \{u_0, u_2,$ $u_4, u_5,$ $u_6, u_7\}$ and $E=\{u_1,$ $u_3, u_8, u_9\}$.
Here $M$ is a \krc, but we can further extend $u_1$ and $u_3$ to $M$,
and come up with a larger \krc. Hence, $M$ is not a maximal \krc.
\end{example}

Since the maximal check algorithm is similar to our advanced enumeration algorithm,
we delay providing the details of this algorithm to Section~\ref{subsec:enumerate}.

\begin{remark}
The early termination technique can be regarded as a lightweight version of the maximal check,
which attempts to terminate the search before a \krc is constructed.
\end{remark}

\subsection{Advanced Enumeration Method}
\label{subsec:enumerate}
\begin{algorithm}[htb]
\SetVline 
\SetFuncSty{textsf}
\SetArgSty{textsf}
\small
\caption{\bf AdvancedEnum($M$, $C$, $E$)}
\label{alg:advancedenu}
\Input
{
   $M :$ chosen vertices set, $C :$ candidate vertices set,
   $E :$ relevant excluded vertices set
}
\Output{$\mathcal{R}:$ maximal \krcs}
\State{Update $C$ and $E$ based on candidate pruning techniques (Theorem~\ref{the:str_prune} and Theorem~\ref{the:sim_prune})}
\label{alg:enu_cand_prune}
\State{\textbf{Return} \textbf{If} current search can be terminated~(Theorem~\ref{the:early_terminate})}
\label{alg:enu_term_1}
\If{$C=SF(C)$ (Theorem~\ref{the:retain})}
{
	\label{alg:enu_check_sf}
	\State{$M:=$ $M \cup C$}
	\label{alg:enu_check_yes1}
\State{$\mathcal{R}:= \mathcal{R} \cup M$ \textit{If} \textbf{CheckMaximal}($M$, $E$)~(Theorem~\ref{the:maximal_check})}
	\label{alg:enu_check_yes2}
}
\Else
{	
	\State{$u \leftarrow$ a vertex in $C \setminus SF(C)$ (Theorem~\ref{the:retain})}
	\label{alg:enu_sub_1}
	\State{\textbf{AdvancedEnum}($M \cup u$, $C \setminus u$, $E$)}
	\label{alg:enu_sub_2}
	\State{\textbf{AdvancedEnum}($M$, $C \setminus u$, $E \cup u$)}	
	\label{alg:enu_sub_3}
}
\end{algorithm}

In Algorithm~\ref{alg:advancedenu}, we present the pseudo code for our advanced enumeration algorithm which integrates
the techniques proposed in previous sections.
We first apply the candidate pruning algorithm outlined in Section~\ref{subsec:enum_prune}
to eliminate some vertices based on structure/similarity constraints.
Along with $C$, we also update $E$
by including discarded vertices and removing the ones that are not similar to $M$.
Line~\ref{alg:enu_term_1} may then terminate the search based on our early termination rules.
If the condition $C=SF(C)$ holds,
$M \cup C$ is a \krc according to Theorem~\ref{the:retain},
and we can conduct the maximal check (Lines~\ref{alg:enu_check_sf}-\ref{alg:enu_check_yes2}).
Otherwise, Lines~\ref{alg:enu_sub_1}-\ref{alg:enu_sub_3} choose one vertex from $C\setminus SF(C)$
and continue the search following two branches.

\begin{algorithm}[htb]
\SetVline 
\SetFuncSty{textsf}
\SetArgSty{textsf}
\small
\caption{\bf CheckMaximal($M$, $C$)}
\label{alg:checkMaximal}
\Input
{
   $M :$ chosen vertices, $C :$ candidate vertices\\
}
\Output{$isMax :$ true if $M$ is a maximal \krc}
\State{Update $C$ based on similarity and structure constraint}
\If{$M$ is a \krc}
{
	\State{Exit the algorithm with $isMax=false$ \textit{If} $|\mathbf{M^*}| <$  $|M|$}
}
\ElseIf{$|C| > 0$}
{	
 	\State{$u \leftarrow$ a vertex in $C$}
 	\State{\textbf{CheckMaximal}($M \cup u$, $C \setminus u$)}
 	\State{\textbf{CheckMaximal}($M$, $C \setminus u$ )}	
}
\end{algorithm}

\noindent \textbf{Checking Maximal Algorithm.}
According to Theorem~\ref{the:maximal_check}, we need to check whether some of the vertices in $E$
can be included in the current \krc, denoted by $M^*$.
This can be regarded as the process of further exploring the search tree by treating $E$ as candidate $C$ (Line~\ref{alg:enu_check_yes2} of Algorithm~\ref{alg:advancedenu}).
Algorithm~\ref{alg:checkMaximal} presents the pseudo code for our maximal check algorithm.\\

%

To enumerate all the maximal \krcs of $G$, we need to replace the \textbf{NaiveEnum} procedure (Line~\ref{alg:naive_enum_1})
in Algorithm~\ref{alg:naive} using our advanced enumeration method (Algorithm~\ref{alg:advancedenu}).
Moreover, the naive checking maximals process (Line~\ref{alg:naive_check_s}-\ref{alg:naive_check_e}) is not necessary
since checking maximals is already conducted by our enumeration procedure (Algorithm~\ref{alg:advancedenu}). Since the search order for vertices does not affect the correctness, the algorithm correctness can be immediately guaranteed based on above analyses. It takes $O(n_e + n_d)$ times for each search node in the worst case, where $n_e$ and $n_d$ denote the total number of edges and dissimilar pairs in $M \cup C \cup E$.

\section{Finding the Maximum (k,r)-core}
\label{sec:max}

In this section, we first introduce the upper bound based algorithm to find the maximum \krc.
Then a novel ($k,k'$)-core approach is proposed to derive tight upper bound of the \krc size.

\subsection{Algorithm for Finding the Maximum One}
\label{subsec:max_upperbounds}

Algorithm~\ref{alg:findmaximum} presents the pseudo code for finding the maximum \krc,
where $R$ denotes the largest \krc seen so far.
There are three main differences compared to the enumeration algorithm (Algorithm~\ref{alg:advancedenu}).
($i$) Line~\ref{alg:maximum_up} terminates the search if we find the current search is non-promising
based on the upper bound of the core size, denoted by \textit{KRCoreSizeUB}($M$,$C$).
($ii$) We do not need to validate the maximal property.
($iii$) Along with the order of visiting the vertices, the order of the two branches also matters
for quickly identifying large \krcs (Lines~\ref{alg:maximum_next_s}-\ref{alg:maximum_next_e}), which is discussed in Section~\ref{sec:order}.

\vspace{1mm}
\begin{algorithm}[htb]
\SetVline 
\SetFuncSty{textsf}
\SetArgSty{textsf}
\small
\caption{\bf FindMaximum($M$, $C$, $E$)}
\label{alg:findmaximum}
\Input
{
   $M :$ chosen vertices set , $C :$ candidate vertices set,
   $E :$ relevant excluded vertices set\\
}
\Output{$R :$ the largest \krc seen so far}
\State{Update $C$ and $E$; Early terminate if possible}
\If{ $KRCoreSizeUB$($M, C$) $> |R|$}
{
	\label{alg:maximum_up}
	\If{$C = SF(C)$}
	{
        \State{$R := M~\cup~C$}
	}
	\Else
	{
		\State{ $u \leftarrow$ choose a vertex in $C \setminus SF(C)$}
		\label{alg:maximum_next_s}
		\If {Expansion is preferred}
		{
			\State{\textbf{FindMaximum}($M \cup u$, $C \setminus u$, $E$)}
		    \State{\textbf{FindMaximum}($M$, $C \setminus u$, $E \cup u$)}
		}	
		\Else
		{
			\State{\textbf{FindMaximum}($M$, $C \setminus u$, $E \cup u$)}
			\State{\textbf{FindMaximum}($M \cup u$, $C \setminus u$, $E$)}		
			\label{alg:maximum_next_e}
        }
	}
}
\end{algorithm}

To find the maximum \krc in $G$, we need to replace the NaiveEnum procedure (Line~\ref{alg:naive_enum}) in Algorithm~\ref{alg:naive} with the method in Algorithm~\ref{alg:findmaximum},
and remove the naive maximal check section of Algorithm~\ref{alg:naive} (Line~\ref{alg:naive_check_s}-\ref{alg:naive_check_e}).
To quickly find a \krc with a large size, we start the algorithm from the subgraph $S$ which holds the
vertex with the highest degree.
The maximum \krc is identified when Algorithm~\ref{alg:naive} terminates.

\vspace{1mm}
\noindent \textbf{Algorithm Correctness.}
Since Algorithm~\ref{alg:findmaximum} is essentially an enumeration algorithm with an upper bound based pruning technique, the correctness of this algorithm is clear if the $KRCoreSizeUB$($M, C$) at Line~\ref{alg:maximum_up} is calculated correctly.

\vspace{1mm}
\noindent \textbf{Time Complexity.}
As shown in Section~\ref{sec:max:upperbounds}, we can efficiently compute the upper bound of core size in $O(n_e+n_s)$ time
where $n_s$ is the number of similar pairs w.r.t $M$ $\cup C$ $\cup E$. For each search node the time complexity of the maximum algorithm is same as that of the enumeration algorithm.

\subsection{Size Upper Bound of ($k$,$r$)-Core}
\label{sec:max:upperbounds}
We use $R$ to denote the \krc derived from $M \cup C$.
In this way, $|M|+|C|$ is clearly an upper bound of $|R|$.
However, it is very loose because it does not consider the similarity constraint.

Recall that $G'$ denotes a new graph that connects the similar vertices of $V(G)$, called \textit{similarity graph}.
By $J$ and $J'$, we denote the induced subgraph of vertices $M \cup C$ from graph $G$ and the similarity graph $G'$, respectively. Clearly, we have $V(J) = V(J')$.
Because $R$ is a \textit{clique} on the similarity graph $J'$ and the size of a $k$-clique is $k$, we can apply the maximum clique size estimation techniques to $J'$ to derive the upper bound of $|R|$.
Color and $k$-core based methods~\cite{DBLP:conf/icde/YuanQLCZ15} are two state-of-the-art techniques for maximum clique size estimation.

\vspace{1mm}
\noindent \textbf{Color based Upper Bound.} Let $c_{min}$ denote the minimum number of colors to \textit{color} the vertices in the similarity graph $J'$ such that every two adjacent vertices in $J'$ have different colors. Since a $k$-clique needs $k$ number of colors to be \textit{colored}, we have $|R| \leq c_{min}$. Therefore, we can apply graph coloring algorithms to estimate a small $c_{min}$~\cite{garey1976complexity}.

\vspace{1mm}
\noindent \textbf{\textit{k}-Core based Upper Bound.} Let $k_{max}$ denote the maximum $k$ value such that $k$-core of $J'$ is not empty.
Since a $k$-clique is also a ($k$-$1$)-core, this implies that we have $|R| \leq k_{max}+1$.
Therefore, we may apply the existing $k$-core decomposition approach~\cite{DBLP:journals/corr/cs-DS-0310049} to compute the maximal core number (i.e., $k_{max}$) on the similarity subgraph $J'$.

At the first glance, both the structure and similarity constraints are used in the above method
because $J$ itself is a $k$-core (structure constraint) and we consider the $k_{max}$-core of $J'$ (similarity constraint). The upper bound could be tighter by choosing the smaller one from color based upper bound and \kc based upper bound. Nevertheless, we observe that the vertices in $k_{max}$-core of $J'$ may not form a \kc on $J$ since we only have $J$ itself as a \kc. If so, we can consider $k_{max}$-$1$ as a tighter upper bound of $R$. Repeatedly, we have the largest $k_{max}$-$i$ as the upper bound such that the corresponding vertices form a \kc on $J$ and a $(k_{max}$-$i)$-core on $J'$. We formally introduce this $(k,k')$-core based upper bound in the following.

\vspace{1mm}
\noindent \textbf{(\textit{k,k'})-Core based Upper Bound.}
We first introduce the concept of \kkc to produce a tight upper bound of $|R|$.
Theorem~\ref{the:up_bound} shows that we can derive the upper bound for any possible \krc $R$
based on the largest possible $k'$ value, denoted by $k'_{max}$, from the corresponding \kkc.

\begin{definition} \textbf{($\mathbf{k},\mathbf{k'}$)-core.}
Given a set of vertices $U$, the graph $J$ and the corresponding similarity graph $J'$, let $J_U$ and $J'_U$ denote the induced subgraph by $U$ on $J$ and $J'$, respectively. If $deg_{min}(J_U) \geq k$ and $deg_{min}(J'_U) = k'$, $U$ is a \kkc of $J$ and $J'$.
\end{definition}

\begin{theorem}
\label{the:up_bound}
Given the graph $J$, the corresponding similarity graph $J'$, and the maximum \krc $R$ derived from $J$ and $J'$, if there is a \kkc on $J$ and $J'$ with the largest $k'$, i.e., $k'_{max}$, we have $|R|\leq k'_{max}+1$.
\end{theorem}

\begin{proof}
Based on the fact that a \krc $R$ is also a ($k$,$k'$)-core with $k'=|R|-1$ according to the definition of \krc,
the theorem is proven immediately.
\end{proof}
\begin{figure}[t]
\begin{center}
\begin{tabular}[t]{c}
      \subfigure[Graph ($J$)]{
      \includegraphics[width=0.47\columnwidth,height=2.6cm]{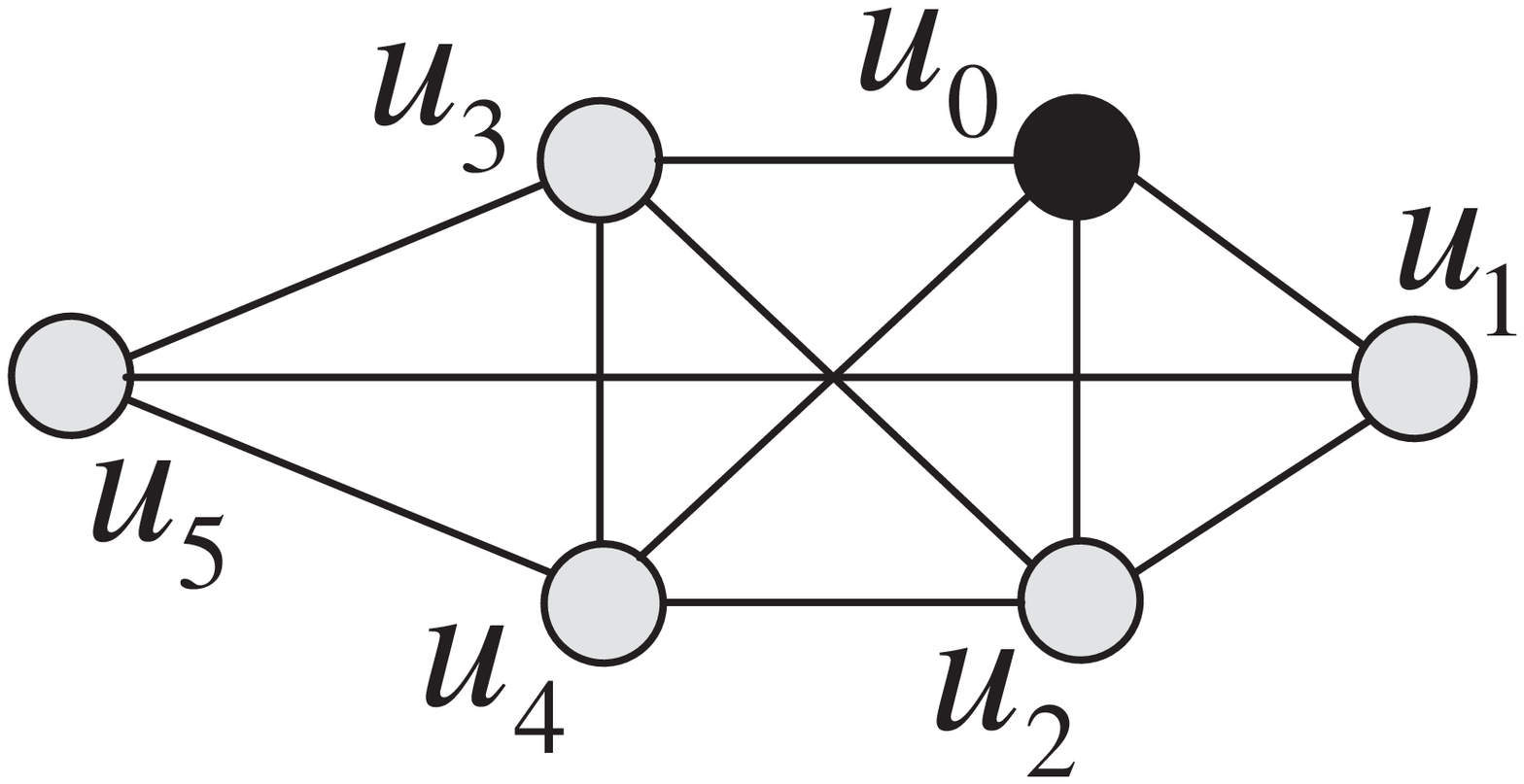}
      \label{fig:max:example structure graph}
     }
      \subfigure[Similarity Graph ($J'$)]{
      \includegraphics[width=0.47\columnwidth,height=2.6cm]{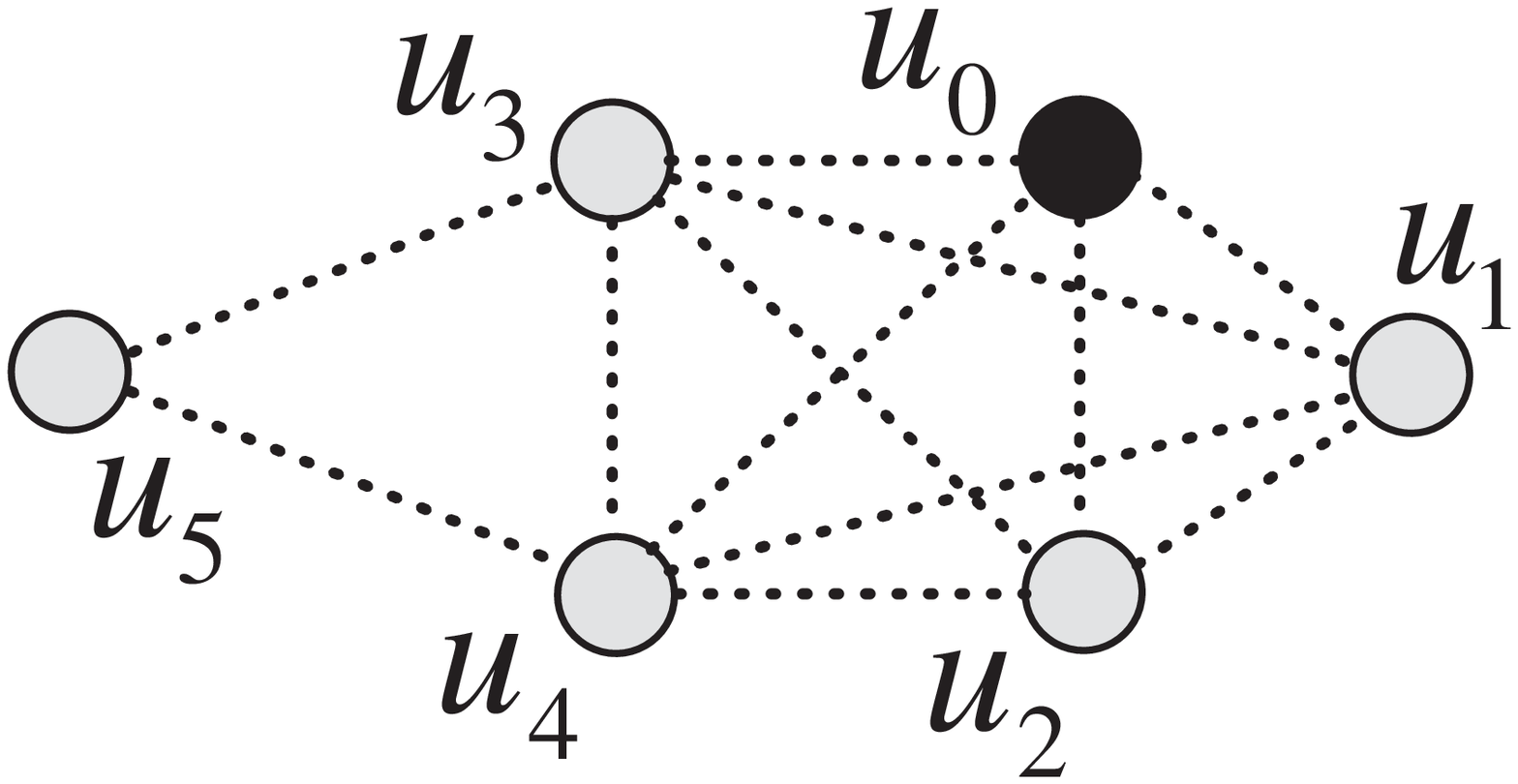}
      \label{fig:max:example similarity graph}
     }
\end{tabular}
\end{center}
\vspace{-7mm}
\caption{\small{Upper Bound Examples}}
\label{fig:max:example}
\end{figure}
\vspace{-3mm}
\begin{example}
In Figure~\ref{fig:max:example}, we have $k=3$, $M=\{u_0\}$ and $C=\{u_1, u_2, u_3, u_4, u_5\}$. Figure~$\ref{fig:max:example structure graph}$ shows the induced subgraph $J$ from $M \cup C$ on $G$
and Figure~\ref{fig:max:example similarity graph} shows the similarity graph $J'$ from $M \cup C$ on the similarity graph $G'$.
We need at least 5 colors to \textit{color} $J'$, so the color based upper bound is 5. By $k$-core decomposition on similarity graph $J'$, we get that the \kc based upper bound is $5$ since $k_{max}=4$.
Regarding the $(k,k')$-core based upper bound,
we can find $k'_{max} = 3$ because there is a $(3,3)$-core on $J$ and $J'$ with four vertices $\{u_0, u_2, u_3, u_4\}$,
and there is no other \kkc with a larger $k'$ than $k'_{max}$. Consequently, the \kkc based upper bound is $4$, which is tighter than $5$.
\end{example}

\subsection{Algorithm for ($k$,$k'$)-Core Upper Bound}
\label{sec:max:upperbounds algorithm}

Algorithm~\ref{alg:kkcorebound} shows the details of the \kkc based upper bound (i.e.,$k'_{max}$) computation, which conducts core decomposition~\cite{DBLP:journals/corr/cs-DS-0310049} on $J'$ with additional update which ensures the corresponding subgraph on $J$ is a \kc.
We use $deg[u]$ and $deg_{sim}[u]$ to denote the degree and similarity degree (i.e., the number of similar pairs from $u$)
of $u$ w.r.t $M \cup C$, respectively.
Meanwhile, $NB[u]$ (resp. $NB_{sim}[u]$) denotes the set of adjacent (resp. similar) vertices of $u$.
The key idea is to recursively mark the $k'$ value of the vertices until we reach the maximal possible value.
Line~\ref{alg:maximum_sort} sorts all vertices based on the increasing order of their similarity degrees.
In each iteration, the vertex $u$ with the lowest similarity degree has already reached its maximal possible $k'$ (Line~\ref{alg:maximum_min}).
Then Line~\ref{alg:maximum_update} invokes the procedure \textbf{KK'coreUpdate} to remove $u$ and
decrease the degree (resp. similarity degree) of its neighbors (resp. similarity neighbors)
at Lines~\ref{alg:maximum_up_sim_s}-\ref{alg:maximum_up_sim_e} (resp. Lines~\ref{alg:maximum_up_deg_s}-\ref{alg:maximum_up_deg_e}).
Note that we need to recursively remove vertices with degree smaller than $k$ (Line~\ref{alg:maximum_up_deg_e}) in the procedure.
At Line~\ref{alg:maximum_order}, we need to reorder the vertices in $H$ since their similarity degree values
may be updated. According to Theorem~\ref{the:up_bound}, $k'+1$ is returned at Line~\ref{alg:maximum_return} as the upper bound of the maximum \krc size.

\begin{algorithm}[htb]
\SetVline 
\SetFuncSty{textsf}
\SetArgSty{textsf}
\small
\caption{\bf KK'coreBound($M$, $C$)}
\label{alg:kkcorebound}
\Input
{
   $M :$ vertices chosen, $C:$ candidate vertices
}
\Output{$k'_{max}:$ the upper bound for the size of the maximum \krc in $M\cup C$ }
\State{$H := $ vertices in $M \cup C$ with increasing order of their similarity degrees}
\label{alg:maximum_sort}
\For{each $u \in H$}
{
    \State{$k':= deg_{sim}(u)$}
	\label{alg:maximum_min}
    \State{\textbf{KK'coreUpdate}($u$, $k'$, $H$)}
	\label{alg:maximum_update}
    \State{reorder $H$ accordingly}
	\label{alg:maximum_order}
}
\Return{$k'+1$}\label{alg:maximum_return}\\

\vspace{2mm}
\textbf{KK'coreUpdate}($u$, $k'$, $H$)\\
\State{Remove $u$ from $H$}
\For{each $v \in NB_{sim}[u] \cap H$}
{
	\label{alg:maximum_up_sim_s}
    ~\If{$deg_{sim}[v] > k'$}
    {
    \State{$deg_{sim}[v] := deg_{sim}[v]-1$}
		\label{alg:maximum_up_sim_e}
    }
}
\For{each $v \in NB[u] \cap H$}
{
	\label{alg:maximum_up_deg_s}
    \State{$deg[v] := deg[v]-1$}
     \If{$deg[v] <  k$}
     {
         \State{\textbf{KK'coreUpdate}($v$, $k'$, $H$)}
		 \label{alg:maximum_up_deg_e}
     }
}
\end{algorithm}

\vspace{1mm}
\noindent \textbf{Time Complexity.}
We can use an array $H$ to maintain the vertices where $H[i]$ keeps the vertices with similarity degree $i$.
Then the sorting of the vertices can be done in $O(|J|)$ time.
The time complexity of the algorithm is $O(n_e + n_s)$, where $n_e$ and $n_s$ denote the number of edges in the graph $J$ and the similarity graph $J'$, respectively.

\vspace{1mm}
\noindent \textbf{Algorithm Correctness.}
Let $k'_{max}(u)$ denote the largest $k'$ value $u$ can contribute to \kkc of $J$.
By $H_{j}$, we represent the vertices $\{ u \}$ with $k'_{max}(u) \geq j$ according to the definition of \kkc.
We then have $H_{j} \subseteq H_{i}$ for any $i < j$.
This implies that a vertex $u$ on $H_i$ with $k'_{max}(u)=i$ will not contribute to $H_j$ with $i < j$.
Thus, we can prove correctness by induction.

\section{Search Order}
\label{sec:order}
 Section~\ref{subsec:ordermodel} briefly introduces some important measurements that should be considered
 for an appropriate visiting order.
 Then we investigate the visiting orders in three algorithms:
 finding the maximum \krc (Algorithm~\ref{alg:findmaximum}), advanced maximal \krc enumeration (Algorithm~\ref{alg:advancedenu}) and maximal check (Algorithm~\ref{alg:checkMaximal})
 at Section~\ref{subsec:maxorder}, Section~\ref{subsec:enuorder}, and Section~\ref{subsec:maxcheckorder}, respectively.
\subsection{Important Measurements}
\label{subsec:ordermodel}
In this paper, we need to consider two kinds of search orders:
($i$) the vertex visiting order: the order of which vertex is chosen from candidate set $C$
and ($ii$) the branch visiting order: the order of which branch goes first (expand first or shrink first).
It is difficult to find simple heuristics or cost functions for two problems studied in this paper
because, generally speaking, finding a maximal/maximum \krc can be regarded as an optimization problem
with two constraints.
On one hand, we need to reduce the number of dissimilar pairs to satisfy the similarity constraint, which implies eliminating a considerable number of vertices from $C$.
On the other hand, the structure constraint and the maximal/maximum property favors a larger number of edges (vertices)
in $M \cup C$; that is, we prefer to eliminate fewer vertices from $C$.

To accommodate this, we propose three measurements where $M'$ and $\mathbf{C'}$ denote the updated $M$ and $C$ after a chosen vertex is extended to $M$ or discarded.
\begin{itemize}
\item \noindent $\Delta_1:$ the change of number of dissimilar pairs, where
\begin{equation}
\label{eqn:delta1}
\Delta_1 = \frac{DP(C) - DP(\mathbf{C'})}{DP(C)}
\end{equation}
Note that we have $DP(u, M \cup C)=0$ for every $u \in M$
according to the similarity invariant (Equation~\ref{eqn:sim_invariant}).

\item $\Delta_2:$ the change of the number of edges, where
\begin{equation}
\label{eqn:delta2}
\Delta_2 = \frac{ |\mathcal{E}(M\cup C)| - |\mathcal{E}(M'\cup \mathbf{C'})| }{|\mathcal{E}(M\cup C)|}
\end{equation}
Recall that $|\mathcal{E}(V)|$ denote the number of edges in the induced graph from the vertices set $V$.

\item $deg$($u,M \cup C$): Degree. We also consider the degree of the vertex as it may reflect its importance.
In our implementation, we choose the vertex with highest degree at the initial stage (i.e., $M=\emptyset$).
\end{itemize}


\subsection{Finding the Maximum \krc}
\label{subsec:maxorder}

Since the size of the largest \krc seen so far is critical to reduce the search space,
we aim to quickly identify the \krc with larger size.
One may choose to carefully discard vertices such that the number of edges in $M$ is reduced slowly
(i.e., only prefer smaller $\Delta_2$ value). However, as shown in our empirical study,
this may result in poor performance because it usually takes many search steps to satisfy the structure constraint.
Conversely, we may easily fall into the trap of finding \krcs with small size
if we only insist on removing dissimilar pairs (i.e., only favor larger $\Delta_1$ value).

In our implementation, we use a cautious greedy strategy where a parameter $\lambda$ is used
to make the trade-off.
In particular, we use $\lambda\Delta_1-\Delta_2$ to measure the suitability of a branch for each vertex in $C\setminus SF(C)$. In this way, each candidate has two scores.
The vertex with the highest score is then chosen and its branch with higher score is explored first (Line~\ref{alg:maximum_next_s}-\ref{alg:maximum_next_e} in Algorithm~\ref{alg:findmaximum}).

For time efficiency, we only explore vertices within two hops from the candidate vertex
when we compute its $\Delta_1$ and $\Delta_2$ values.
It takes $O(n_c\times(d_1^2 + d_2^2))$ time where $n_c$ denote the number of vertices in $C\setminus SF(C)$, and $d_1$ (resp. $d_2$) stands for the average degree of the vertices in $J$ (resp. $J'$).

\subsection{Enumerating All Maximal \krcs}
\label{subsec:enuorder}
The ordering strategy in this section differs from finding the maximum in two ways.

($i$) We observe that $\Delta_1$ has much higher impact than $\Delta_2$ in the enumeration problem,
so we adopt the $\Delta_1$-then-$\Delta_2$ strategy;
that is, we prefer the larger $\Delta_1$, and the smaller $\Delta_2$ is considered if there is a tie.
This is because the enumeration algorithm does not prefer \krc with very large size
since it eventually needs to enumerate all maximal \krcs.
Moreover, by the early termination technique proposed in Section~\ref{subsec:enum_term},
we can avoid exploring many non-promising subtrees that were misled by the greedy heuristic.

($ii$) We do not need to consider the search order of two branches because both must be explored eventually.
Thus, we use the score summation of the two branches to evaluate the suitability of a vertex.
The complexity of this ordering strategy is the same as that in Section~\ref{subsec:maxorder}.

\subsection{Checking Maximal}
\label{subsec:maxcheckorder}
The search order for checking maximals is rather different than the enumeration and maximum algorithms.
Towards the checking maximals algorithm, it is cost-effective to find a \textit{small} \krc which fully contains the candidate \krc.
To this end, we adopt a short-sighted greedy heuristic.
In particular, we choose the vertex with the largest degree and the \eeb is always preferred
as shown in Algorithm~\ref{alg:checkMaximal}.
By continuously maintaining a priority queue, we fetch the vertex with the highest degree in $O(\log |C|)$ time.
%
%
%
%

\section{Performance Evaluation}
\label{sec:exp}
This section evaluates the effectiveness and efficiency of our algorithms through comprehensive experiments.

\subsection{Experimental Setting}
\label{subsec:exp setting}

\begin{table}[htb]
\vspace{1mm}
\small
  \centering
    \begin{tabular}{|p{0.18\columnwidth}|p{0.72\columnwidth}|}
      \hline
      \textbf{Algorithm}   & \textbf{Description}                   \\ \hline \hline
      $\mathbf{Clique+}$  &  The advanced clique-based algorithm proposed in Section~\ref{sec:disc}, using the clique and \kc computation algorithms in~\cite{wang2013redundancy} and~\cite{DBLP:journals/corr/cs-DS-0310049}, respectively. The source code for maximal clique enumeration was downloaded from \url{http://www.cse.cuhk.edu.hk/~jcheng/publications.html}. \\ \hline
      $\mathbf{BasicEnum}$  &  The basic enumeration method proposed in Algorithm~\ref{alg:naive} including the structure and similarity constraints based pruning techniques (Theorems~\ref{the:str_prune} and \ref{the:sim_prune} in Section~\ref{subsec:enum_prune})\footnote{We do not specifically evaluate the structure/similarity based candidate pruning techniques because they are indispensable for the baseline algorithm.}. The best search order ($\Delta_1$-then-$\Delta_2$, in Section~\ref{subsec:enuorder}) is applied. \\ \hline
      $\mathbf{BasicMax}$  &  The algorithm proposed in Section~\ref{subsec:max_upperbounds} with the upper bound replaced by a naive one: $|M|+|C|$. The best search order is applied ($\lambda\Delta_1-\Delta_2$, in Section~\ref{subsec:maxorder}). \\ \hline
      $\mathbf{AdvEnum}$  &  The advanced enumeration algorithm proposed in Section~\ref{subsec:enumerate} that applies all advanced pruning techniques including: candidate size reduction (Theorems~\ref{the:str_prune},~\ref{the:sim_prune} and~\ref{the:retain} in Section~\ref{subsec:enum_prune}), early termination (Theorem~\ref{the:early_terminate} in Section~\ref{subsec:enum_term}) and checking maximals (Theorem~\ref{the:maximal_check} in Section~\ref{subsec:enum_check}). Moreover, the best search order is used ($\Delta_1$-then-$\Delta_2$, in Section~\ref{subsec:enuorder}). \\ \hline
      $\mathbf{AdvMax}$  &  The advanced finding maximum \krc algorithm proposed in Section~\ref{subsec:max_upperbounds} including \kkc based upper bound technique (Algorithm~\ref{alg:kkcorebound}). Again, the best search order is applied ($\lambda\Delta_1-\Delta_2$, in Section~\ref{subsec:maxorder}). \\ \hline
\end{tabular}
\vspace{-3.5mm}
\caption{Summary of Algorithms}
\label{tb:algorithms}
\end{table}

\vspace{1mm}
\noindent \textbf{Algorithms.} To the best of our knowledge, there are no existing works that investigate the problem of maximal \krc enumeration or finding the maximum \krc.
In this paper, we implement and evaluate 2 baseline algorithms, 2 advanced algorithms and the clique-based algorithm which are described in Table~\ref{tb:algorithms}.

Since the naive method in Section~\ref{sec:naive} is extremely slow even on a small graph,
we employ \basicenum and \basicmax as the baseline algorithms in the empirical study for the problem of enumerating all maximal \krcs and finding the maximum \krc, respectively.

\begin{table}
\vspace{2mm}
\small
  \centering
    \begin{tabular}{|l|l|l|l|l|}
      \hline
      \textbf{Dataset}  & \textbf{Nodes}  & \textbf{Edges} & $d_{avg}$ & $d_{max}$\\ \hline \hline

      \texttt{Brightkite}  &  58,228 & 194,090 & 6.7 & 1098\\ \hline

      \texttt{Gowalla}  & 196,591 & 456,830 & 4.7 & 9967\\ \hline

      \texttt{DBLP} &  1,566,919 & 6,461,300 & 8.3 & 2023\\ \hline

      \texttt{Pokec}  & 1,632,803 & 8,320,605 & 10.2 & 7266\\ \hline
%
%
%
%

    \end{tabular}
\vspace{-2mm}
\caption{Statistics of Datasets}
\label{tb:datasets}
\end{table}

\vspace{1mm}
\noindent \textbf{Datasets.}
Four real datasets are used in our experiments. The original data of \texttt{DBLP} was downloaded from \url{http://dblp.uni-trier.de/} and the remaining three datasets were downloaded from \url{http://snap.stanford.edu/}.
In \texttt{DBLP}, we consider each author as a vertex with attribute of counted 'attended conferences' and 'published journals' list.
There is an edge for a pair of authors if they have at least one co-authored paper.
We use \textit{Weighted Jaccard Similarity} between the corresponding attributes (counted conferences and journals) to measure the similarity between two authors.
In \texttt{Pokec}, we consider each user to be a vertex with personal interests.
We use \textit{Weighted Jaccard Similarity} as the similarity metric.
And there is an edge between two users if they are friends and
In \texttt{Gowalla} and \texttt{Brightkite}, we consider each user as a vertex along with his/her location information.
The graph is constructed based on friendship information.
We use \textit{Euclidean Distance} between two locations to measure the similarity between two users.
Table~\ref{tb:datasets} shows the statistics of the four datasets.

\vspace{1mm}
\noindent \textbf{Parameters.} We conducted experiments using different settings of $k$
and $r$. We set reasonable positive integers for $k$, which varied from $3$ to $15$.
In \texttt{Gowalla} and \texttt{Brightkite}, we used Euclidean distance as the distance threshold $r$, ranging from $1$ km to $500$ km.
The pairwise similarity distributions are highly skewed in \texttt{DBLP} and \texttt{Pokec}. Thus, we used the thousandth of the pairwise similarity distribution in decreasing order which grows from top $1$\textperthousand~to~top~$15$\textperthousand (i.e., the similarity threshold value drops).
Regarding the search orders of the \advmax and \basicmax algorithms, we set $\lambda$ to $5$ by default.

All programs were implemented in standard C++ and compiled with G++ in Linux.
All experiments were performed on a machine with an Intel Xeon 2.3GHz CPU and a Redhat Linux system.
We evaluate the performance of an algorithm by its running time.
To better evaluate the difference between the algorithms, we set the time cost to \textbf{INF} if an algorithm did not terminate within one hour. We also report the number of maximal \krcs and their average and maximum sizes.

\subsection{Effectiveness}
\label{subsec:exp effect}

\begin{figure}[t]
\begin{center}
      \subfigure[Enumeration]{
      \includegraphics[width=0.95\columnwidth,height=3.5cm]{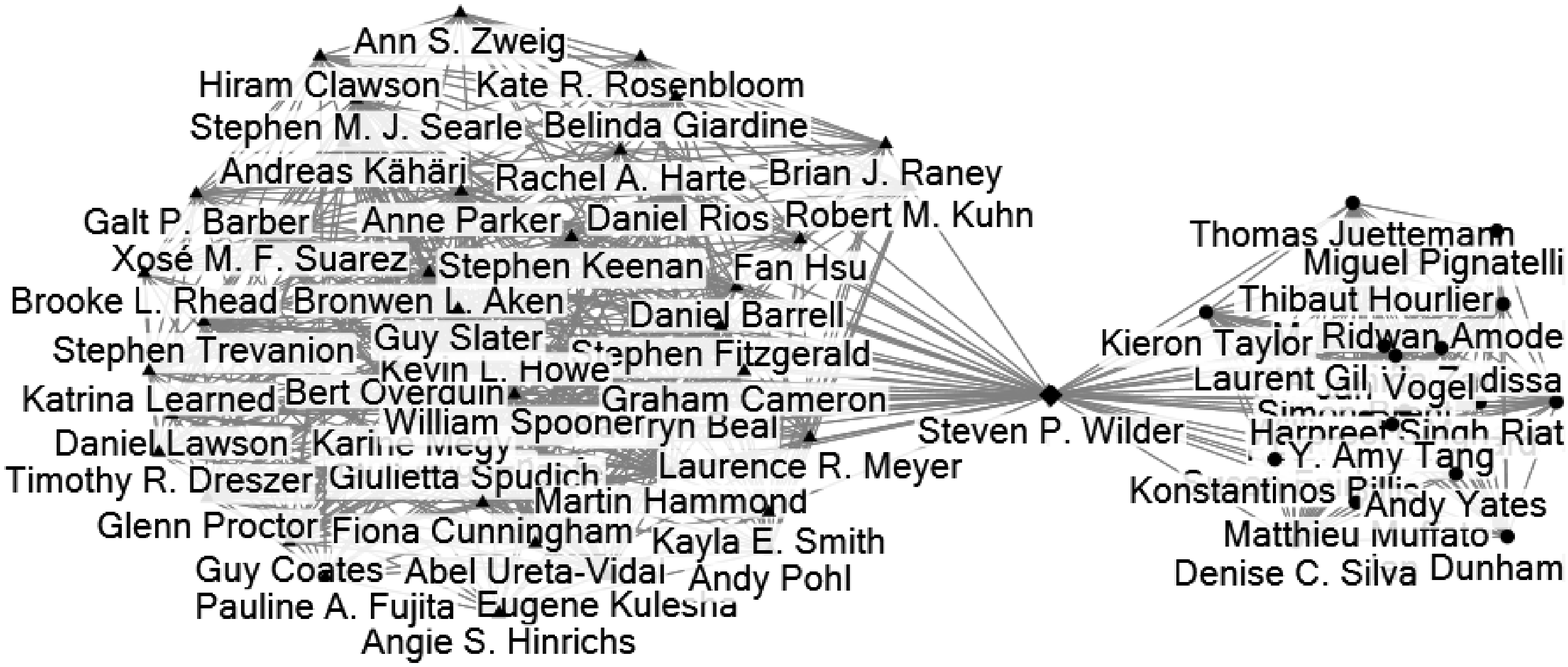}}
      \vspace{-2mm}

      \subfigure[Maximum]{
      \includegraphics[width=0.8\columnwidth,height=4cm]{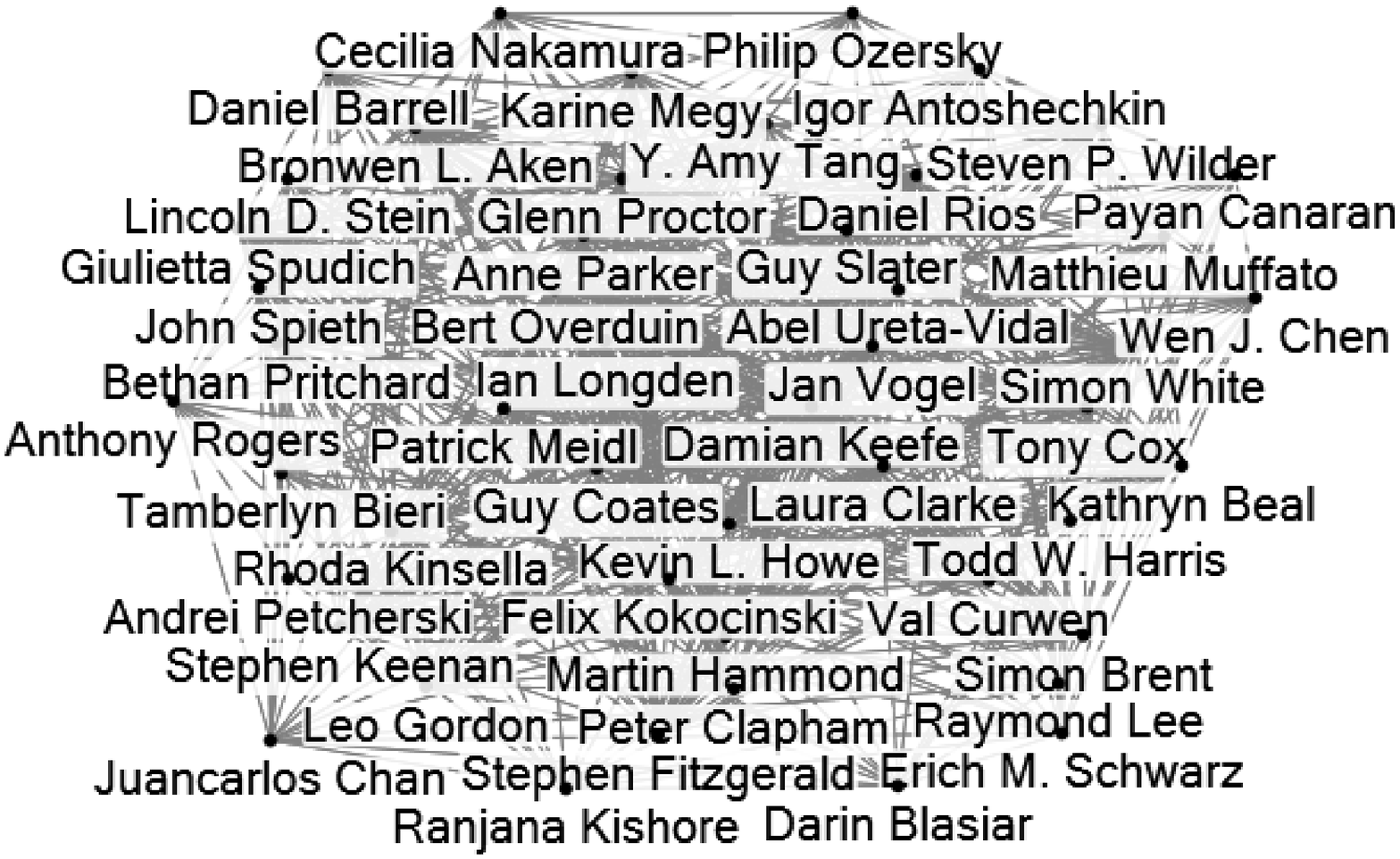}}
\end{center}
\vspace{-8mm}
\caption{\small{Case Study on DBLP (k=15, r=top 3\textperthousand)}}
\vspace{-2mm}
\label{fig:exp:effectiveness}
\end{figure}
\begin{figure}[t]
\begin{center}
      \includegraphics[width=0.8\columnwidth,height=4cm]{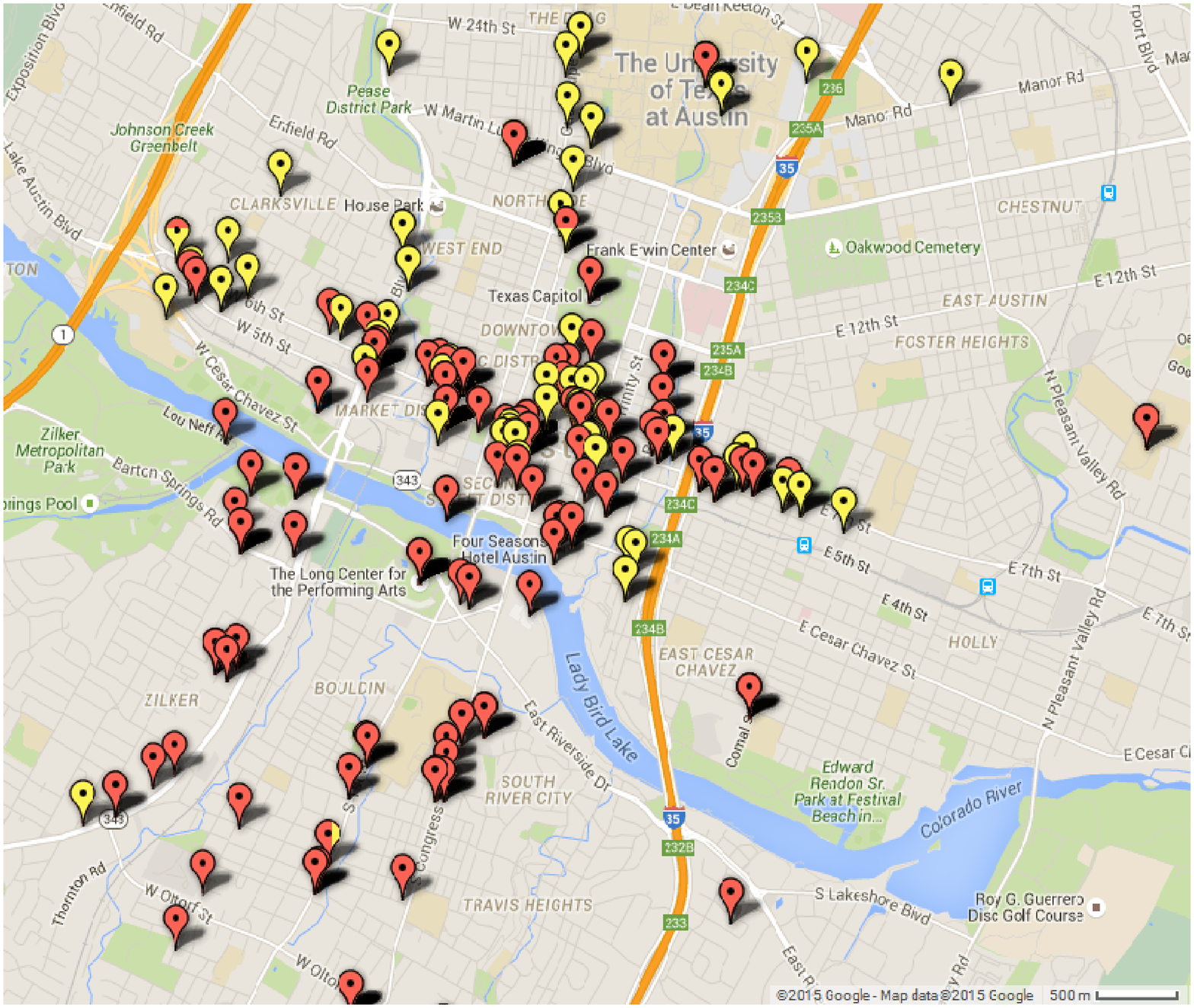}
\end{center}
\vspace{-6mm}
\caption{\small{Case Study on Gowalla (k=10, r=10km)}}
\label{fig:exp:effectiveness1}
\end{figure}

We conducted two case studies on \texttt{DBLP} and \texttt{Gowalla} to demonstrate the effectiveness of our \krc model.
We observe that, compared to \kc, \krc enables us to find more valuable information with the additional similarity constraint
on the vertex attribute.

\vspace{1mm}
\noindent \textbf{DBLP.}
Figure~\ref{fig:exp:effectiveness}(a) and (b) show two examples of \texttt{DBLP} with $k=15$ and r= $3$\textperthousand
\footnote{To avoid the noise, we enforce that there are at least three co-authored papers between two connected authors in the case study.
}. In Figure~\ref{fig:exp:effectiveness}(a), all authors come from the same $k$-core
based on their co-authorship information alone (their structure constraint).
While there are two \krcs with one common author named Steven P. Wilder, if we also consider their research background (their similarity constraint). We find from searching the internet that a large number of the authors in the left \krc are bioinformaticians
from the European Bioinformatics Institute (EBI), while many of the authors in the right \krc are from
the Wellcome Trust Centre. Moreover, it turns out that Dr. Wilder got his Ph.D. from the Wellcome Trust Centre for Human Genetics, University of Oxford in $2007$, and has worked at EBI ever since.
Figure~\ref{fig:exp:effectiveness}(b) depicts the maximum \krc of DBLP with $49$ authors.
We find that they have intensively co-authored many papers related to a project named \textit{Ensembl} (\url{http://www.ensembl.org/index.html}), which is one of the well known genome browsers.
It is very interesting that, although the size of maximum \krc changes when we vary the values of $k$ and $r$,
the authors remaining in the maximum \krc are closely related to the project.

\vspace{1mm}
\noindent \textbf{Gowalla.}
Figure ~\ref{fig:exp:effectiveness1} illustrates a set of \texttt{Gowalla} users who are from the same \kc with $k=10$.
By setting $r$ to $10$ km, two groups of users emerge, each of which is a maximal \krc, and we cannot identify them by structure constraint or similarity constraint alone.
We observe that the maximum \krc in \texttt{Gowalla} always appears at Austin when $k \geq 6$.
Then we realize that this is because the headquarters of \texttt{Gowalla} is located in Austin.

\begin{figure}[htb]

\begin{center}
    \includegraphics[width=1.0\columnwidth]{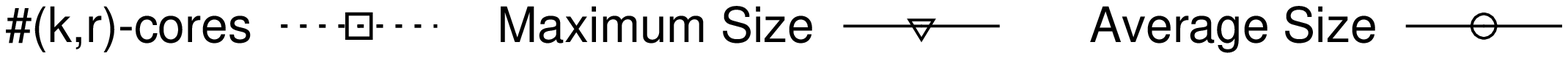}\vspace{-2mm}

    \subfigure[Gowalla, k=5]{
    \includegraphics[width=0.47\columnwidth,height=2.5cm]{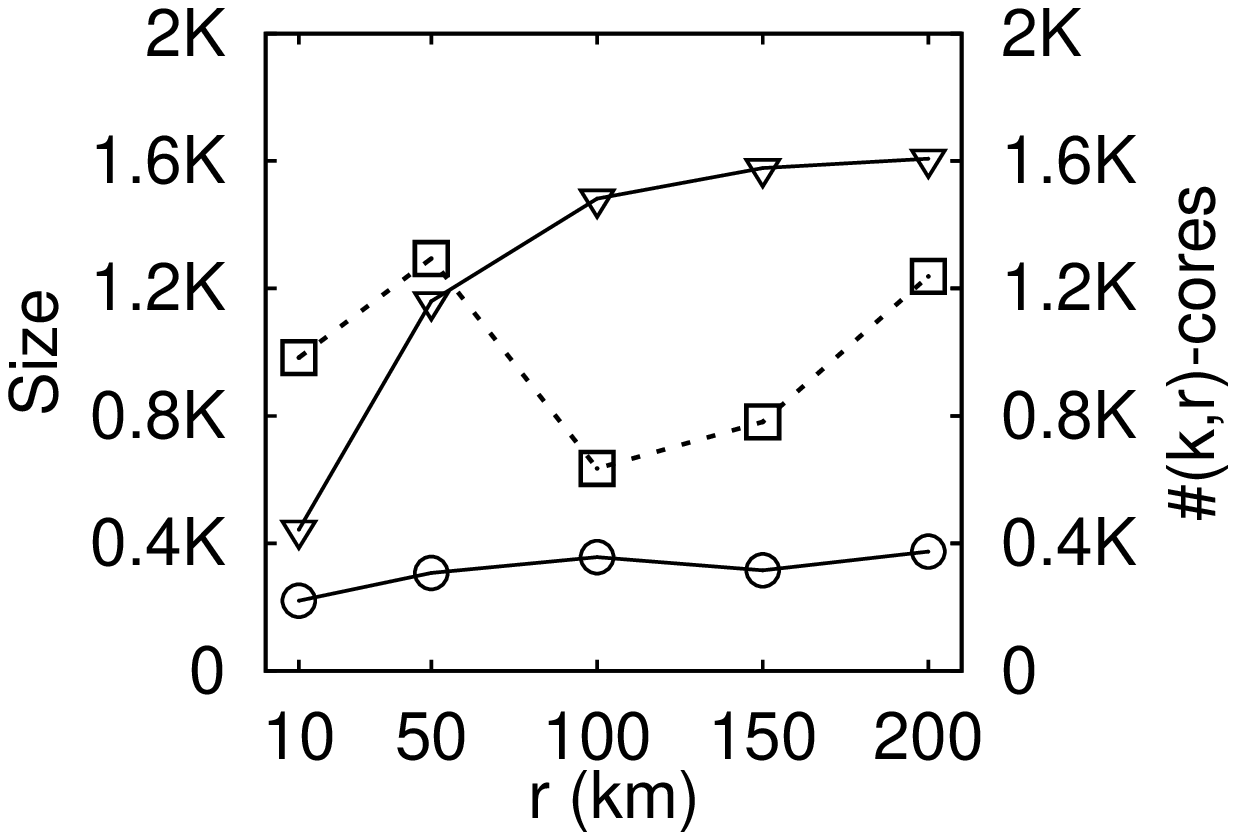}}
    \subfigure[DBLP, r=top 3\textperthousand]{
    \includegraphics[width=0.47\columnwidth,height=2.5cm]{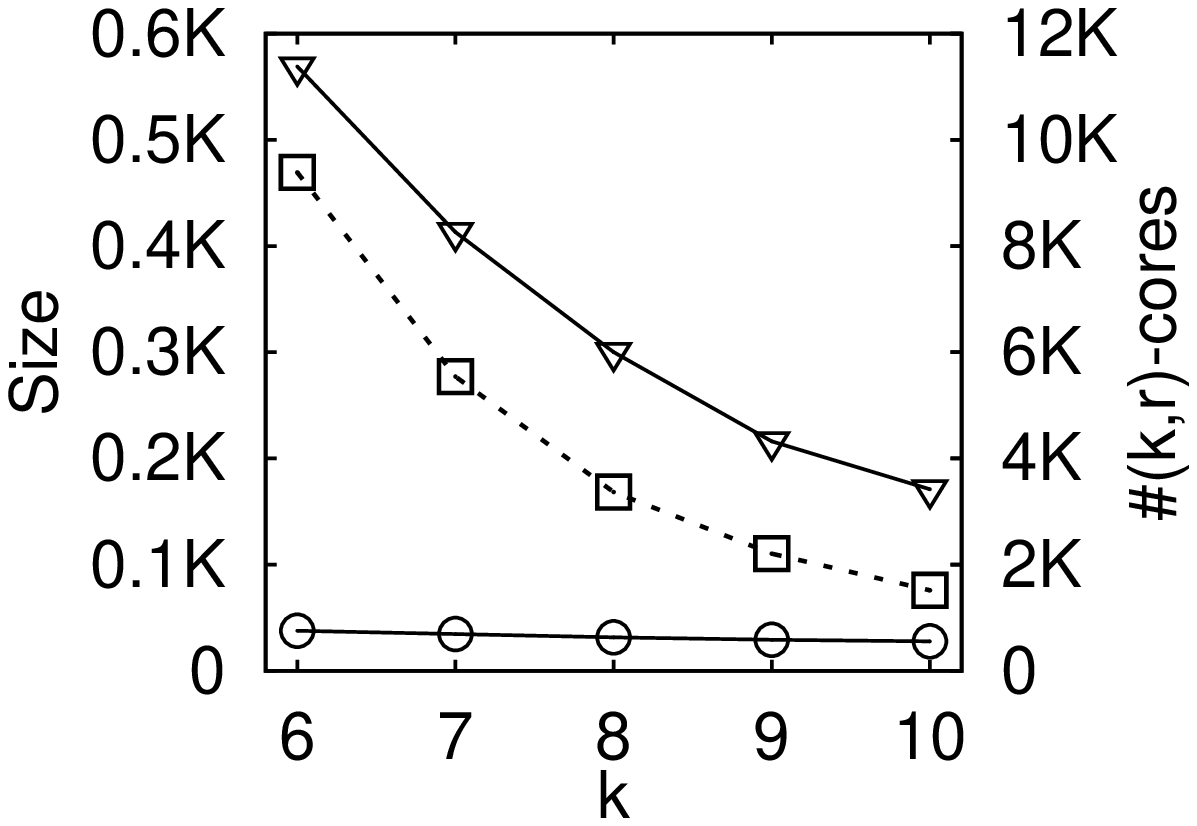}
    }
\end{center}
\vspace{-7mm}
\caption{\small{\krc Statistics}}
\label{fig:exp:corestati}
\end{figure}

We also report the number of \krcs, the average size and maximum size of \krcs on \texttt{Gowalla} and \texttt{DBLP}.
Figure~\ref{fig:exp:corestati}(a) and (b) show that both maximum size of \krcs and the number of \krcs
are much more sensitive to the change of $r$ or $k$ on the two datasets, compared to the average size.

\subsection{Efficiency}
\label{subsec:exp efficiency}
In this section, we evaluate the efficiency of the techniques proposed in this paper and report the time costs of the algorithms.

\begin{figure}
	\begin{minipage}{\columnwidth}
		\centering
		\includegraphics[width=1.0\columnwidth]{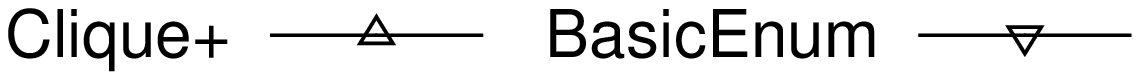}%
	\end{minipage}\vspace{-2mm}
	
	\subfigure[Gowalla, k=5]{
    	\label{fig:exp:tuning_gowalla} 
    	\includegraphics[width=0.48\columnwidth, height=2.5cm]{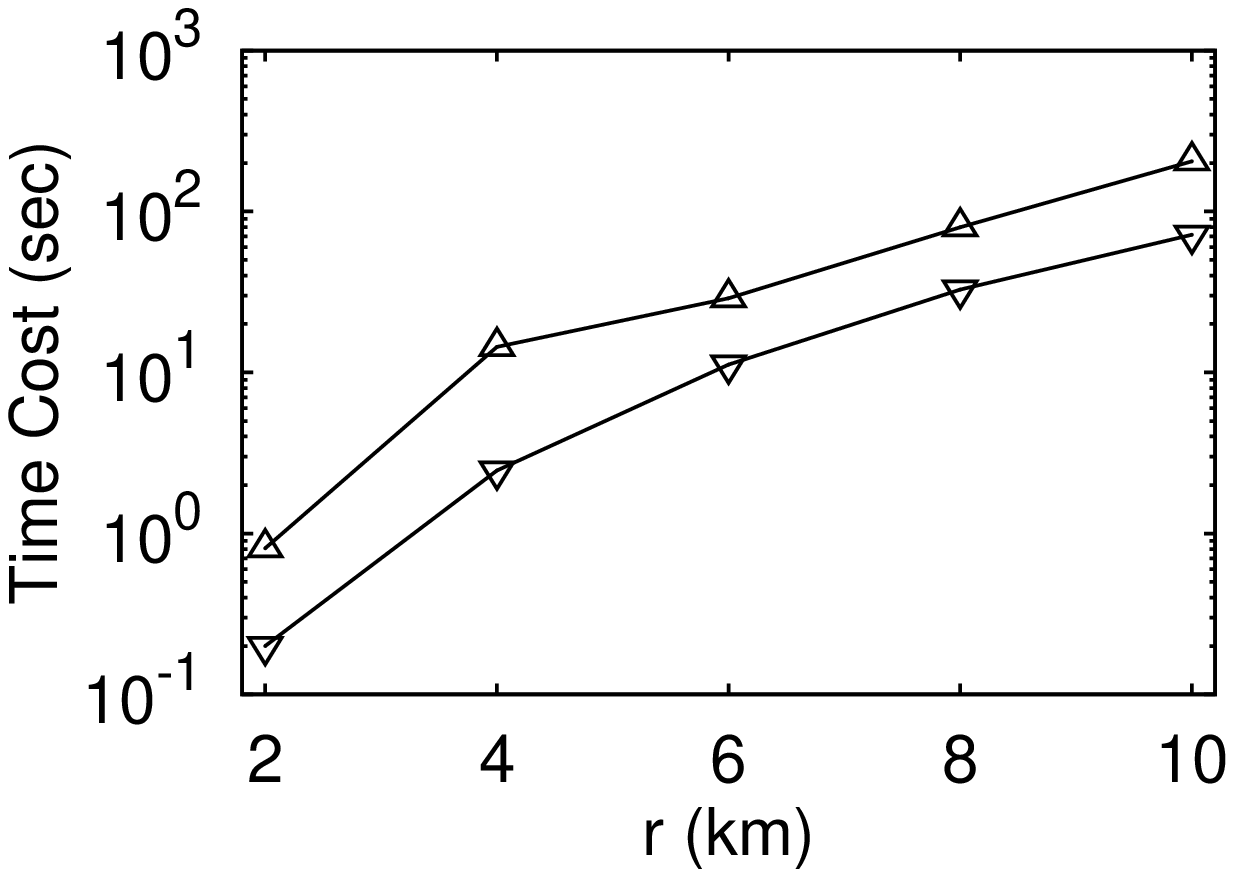}}
  	\subfigure[DBLP, r=top 3\textperthousand]{
    	\label{fig:exp:tuning_dblp} 
    	\includegraphics[width=0.48\columnwidth, height=2.5cm]{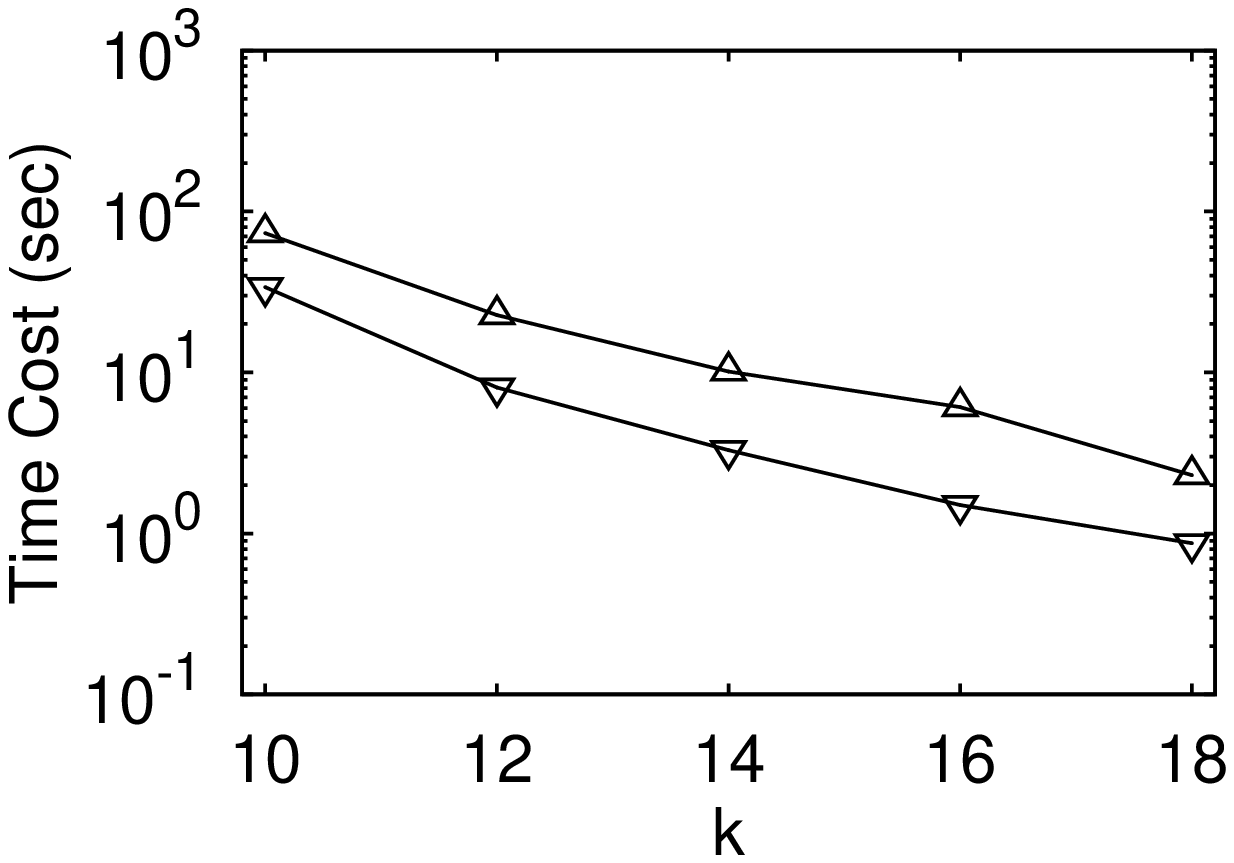}}
\vspace{-5mm}
\caption{\small Evaluate Clique-based Method}
\vspace{-3mm}
\label{fig:exp:cliquemethod} 
\end{figure}

\begin{figure}
	\begin{minipage}{\columnwidth}
		\centering
		\includegraphics[width=1.0\columnwidth]{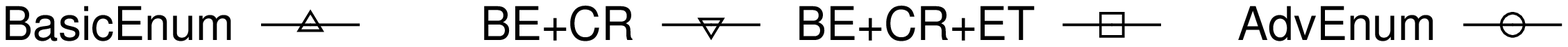}%
	\end{minipage}\vspace{-2mm}
	
	\subfigure[Gowalla, k=5]{
    	\label{fig:exp:tuning_gowalla} 
    	\includegraphics[width=0.48\columnwidth, height=2.5cm]{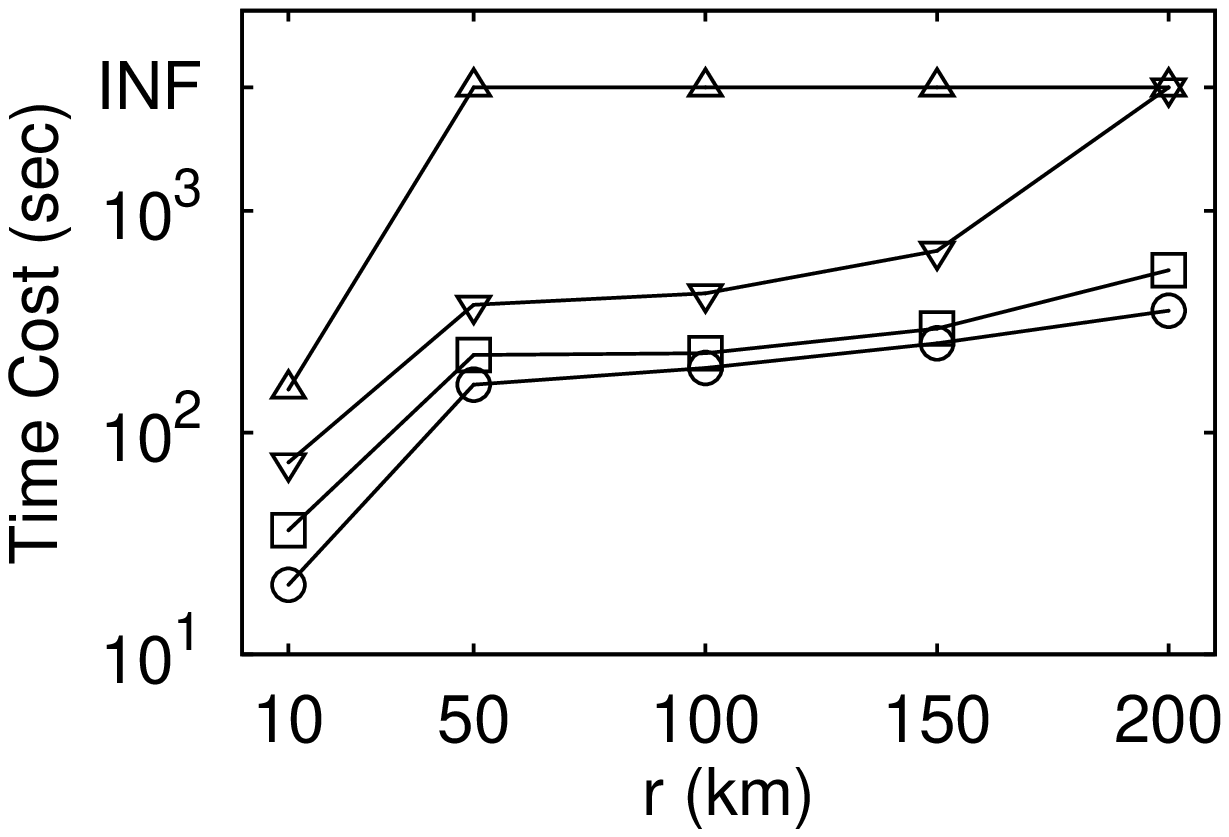}}
  	\subfigure[DBLP, r=top 3\textperthousand]{
    	\label{fig:exp:tuning_dblp} 
    	\includegraphics[width=0.48\columnwidth, height=2.5cm]{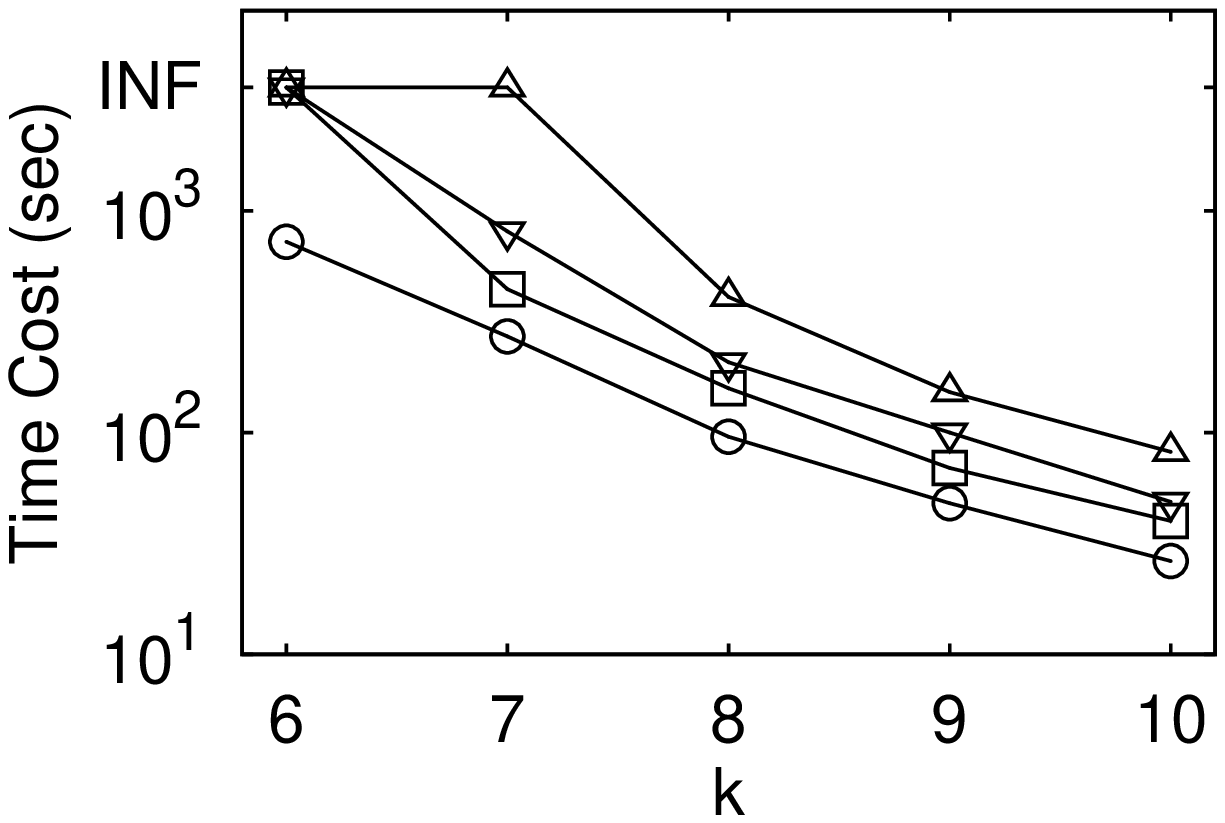}}
\vspace{-5mm}
\caption{\small Evaluate Pruning Techniques}
\label{fig:exp:tuning} 
\end{figure}

\vspace{1mm}
\noindent \textbf{Evaluating the Clique-based Method.}
In Figure~\ref{fig:exp:cliquemethod}, we evaluate the time cost of the maximal \krc enumeration for
$\textbf{Clique+}$ and \basicenum on the \texttt{Gowalla} and \texttt{DBLP} datasets.
In Figure~\ref{fig:exp:tuning_gowalla}, we fix the structure constraint $k$ to $5$ and vary the similarity threshold $r$ from $2$ km to $10$ km.
Correspondingly in Figure~\ref{fig:exp:tuning_dblp}, we fix $r$ to $3$\textperthousand~and vary $k$ from $18$ to $10$.
\basicenum is typically considered to significantly outperform $\textbf{Clique+}$ because
a large number of of cliques are materialized in the similarity graphs when the clique-based method is employed.
Consequently, we exclude $\textbf{Clique+}$ from the following experiments.

\vspace{1mm}
\noindent \textbf{Evaluating the Pruning Techniques.}
We evaluate the efficiency of our candidate size reduction, early termination and checking maximals techniques on \texttt{Gowalla} and \texttt{DBLP} in Figure~\ref{fig:exp:tuning} by incrementally integrating these techniques into the \basicenum algorithm.
Particularly, $\mathbf{BE}$+$\mathbf{CR}$ represents the \basicenum algorithm with \underline{c}andidate \underline{r}etaining technique (Theorem~\ref{the:retain}).
Then $\mathbf{BE}$+$\mathbf{CR}$+$\mathbf{ET}$ further includes the \underline{e}arly \underline{t}ermination technique (Theorem~\ref{the:early_terminate}).
By integrating the checking maximal technique (Theorem~\ref{the:maximal_check}), it turns to be our \advenum algorithm.
Note that the best search order is used for all algorithms.
The results in Figure~\ref{fig:exp:tuning} confirm that all techniques make a significant contribution to
enhance the performance of our \advenum algorithm.

\begin{figure}
\begin{center}
	\begin{minipage}{\columnwidth}
        \centering
		\includegraphics[width=1.0\columnwidth]{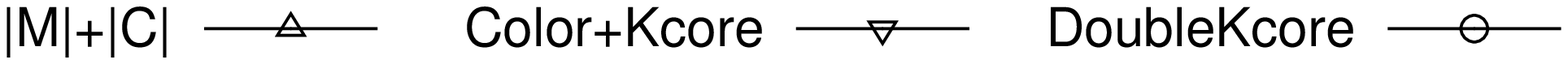}
	\end{minipage}\vspace{-2mm}

    \subfigure[DBLP, k=10]{
    \includegraphics[width=0.47\columnwidth, height=2.5cm]{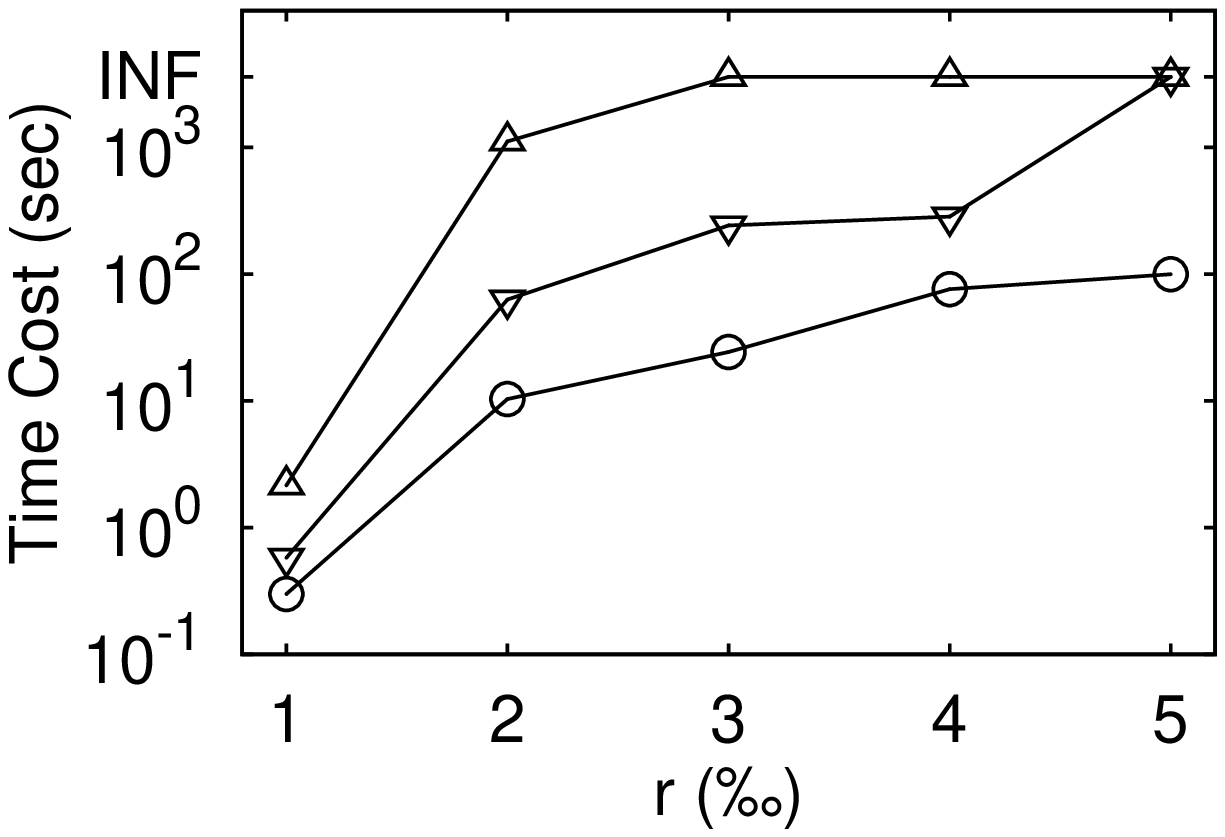}}
    \subfigure[DBLP, r=top 3\textperthousand]{
    \includegraphics[width=0.47\columnwidth, height=2.5cm]{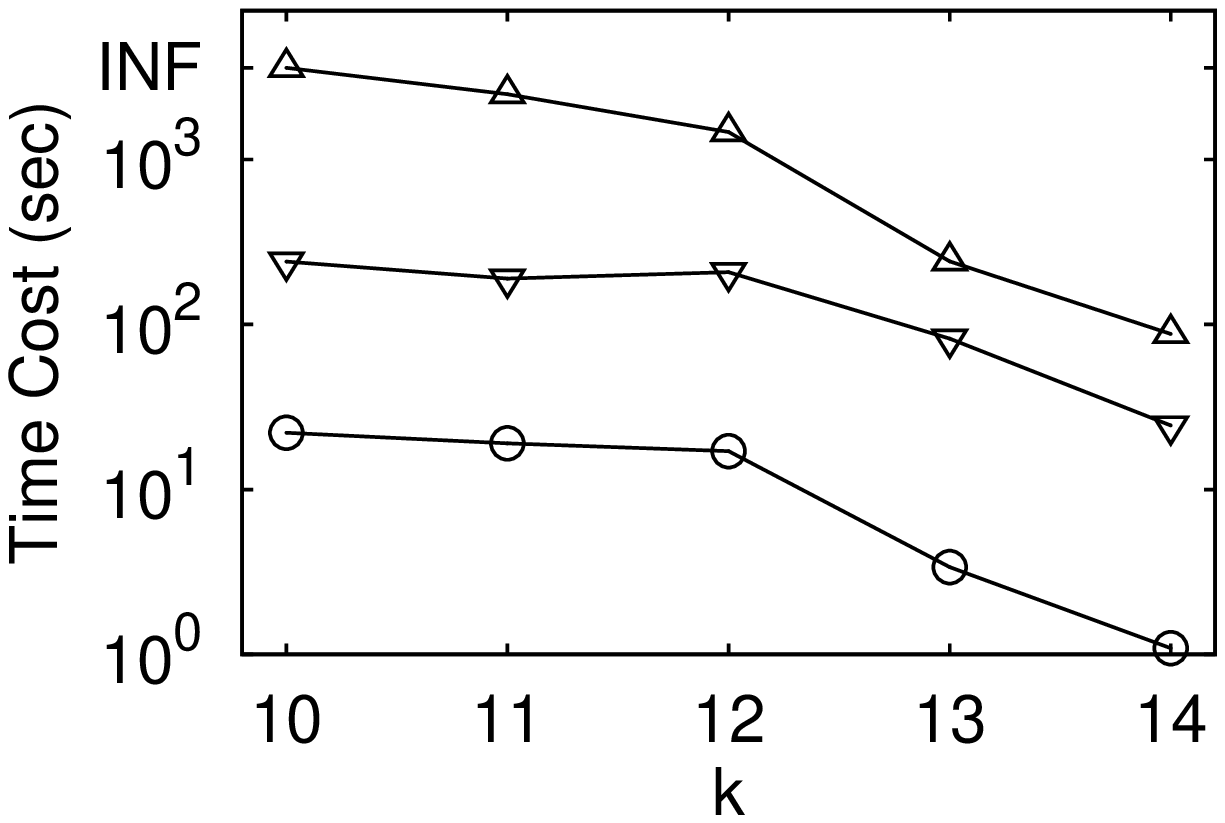}}
\end{center}
\vspace{-7mm}
\caption{\small{Evaluate Upper Bounds}}
\label{fig:exp:upperbounds}
\end{figure}

\vspace{1mm}
\noindent \textbf{Evaluating the Upper Bound Technique.}
Figure~\ref{fig:exp:upperbounds} demonstrates the effectiveness of the \kkc based upper bound technique (Algorithm~\ref{alg:kkcorebound}) on DBLP by varying the values of $r$ and $k$.
In $\mathbf{Color}$+$\mathbf{Kcore}$~\cite{DBLP:conf/icde/YuanQLCZ15}, we used the better upper bound of color and \kc based techniques.
Studies show that $\mathbf{Color}$+$\mathbf{Kcore}$ significantly enhances performance compared to the naive upper bound
($|M|+|C|$).
Nevertheless, our \kkc based upper bound technique beats the color and \kc based techniques by a large margin
because it can better exploit the structure/similarity constraints.
\begin{figure}
\begin{center}
    \subfigure[Maximum(Maxm), $\lambda$]{
    \includegraphics[width=0.48\columnwidth, height=2.5cm]{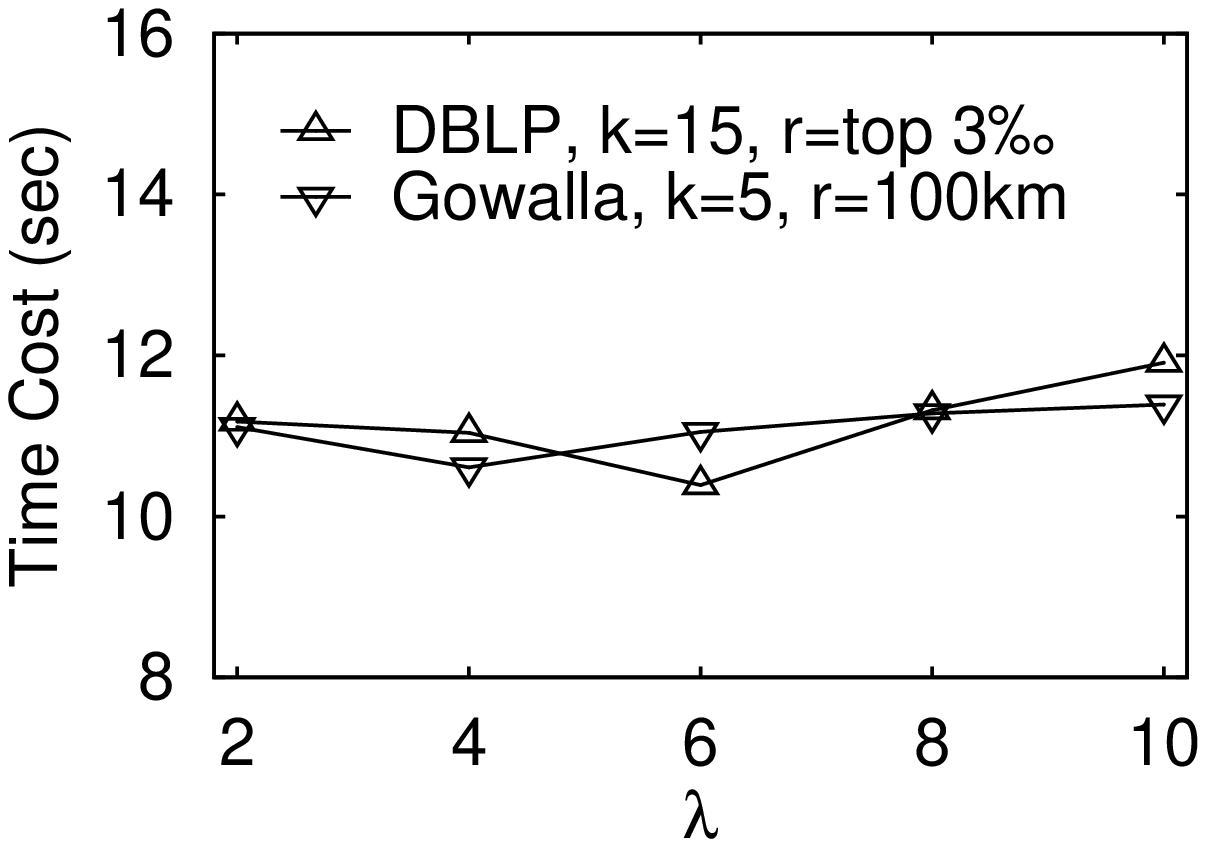}}    
    \subfigure[Maxm, DBLP, r=top 3\textperthousand]{
    \includegraphics[width=0.48\columnwidth, height=2.5cm]{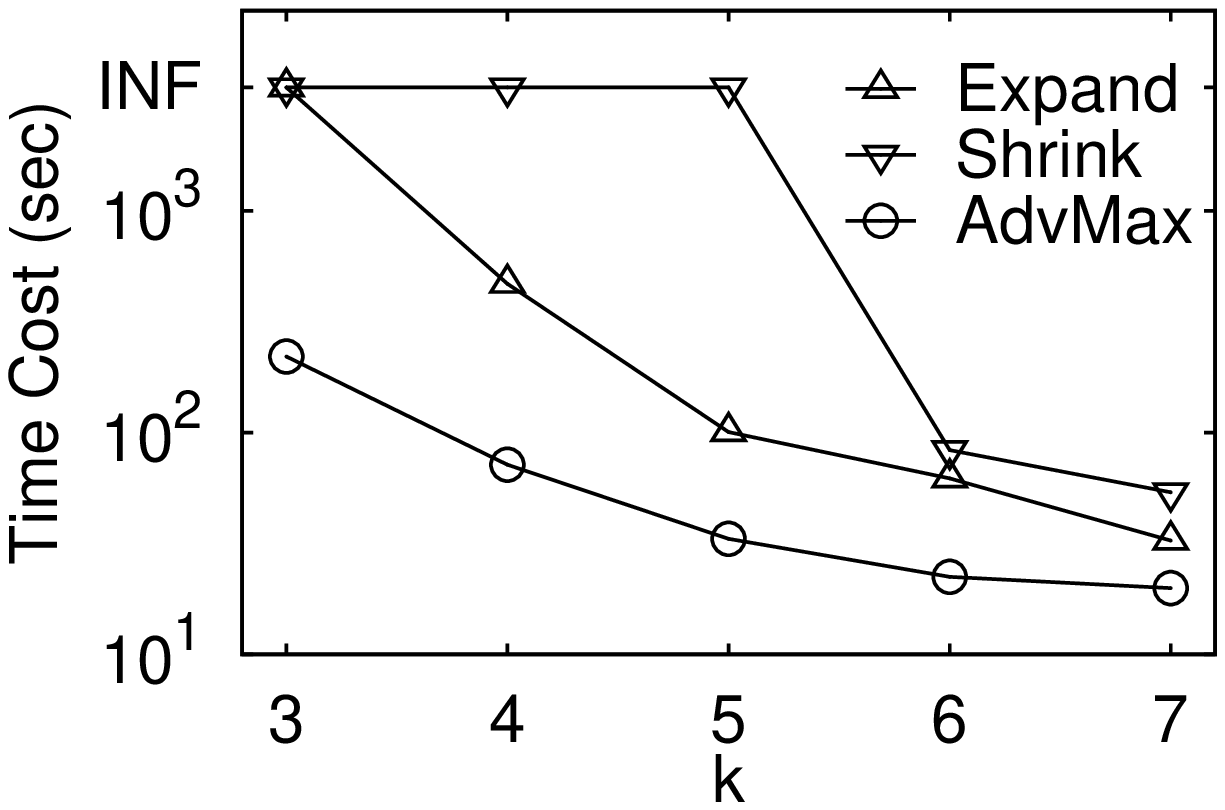}}
    \vspace{-2mm}

    \subfigure[Maxm, DBLP, r=top 3\textperthousand]{
    \includegraphics[width=0.48\columnwidth, height=2.5cm]{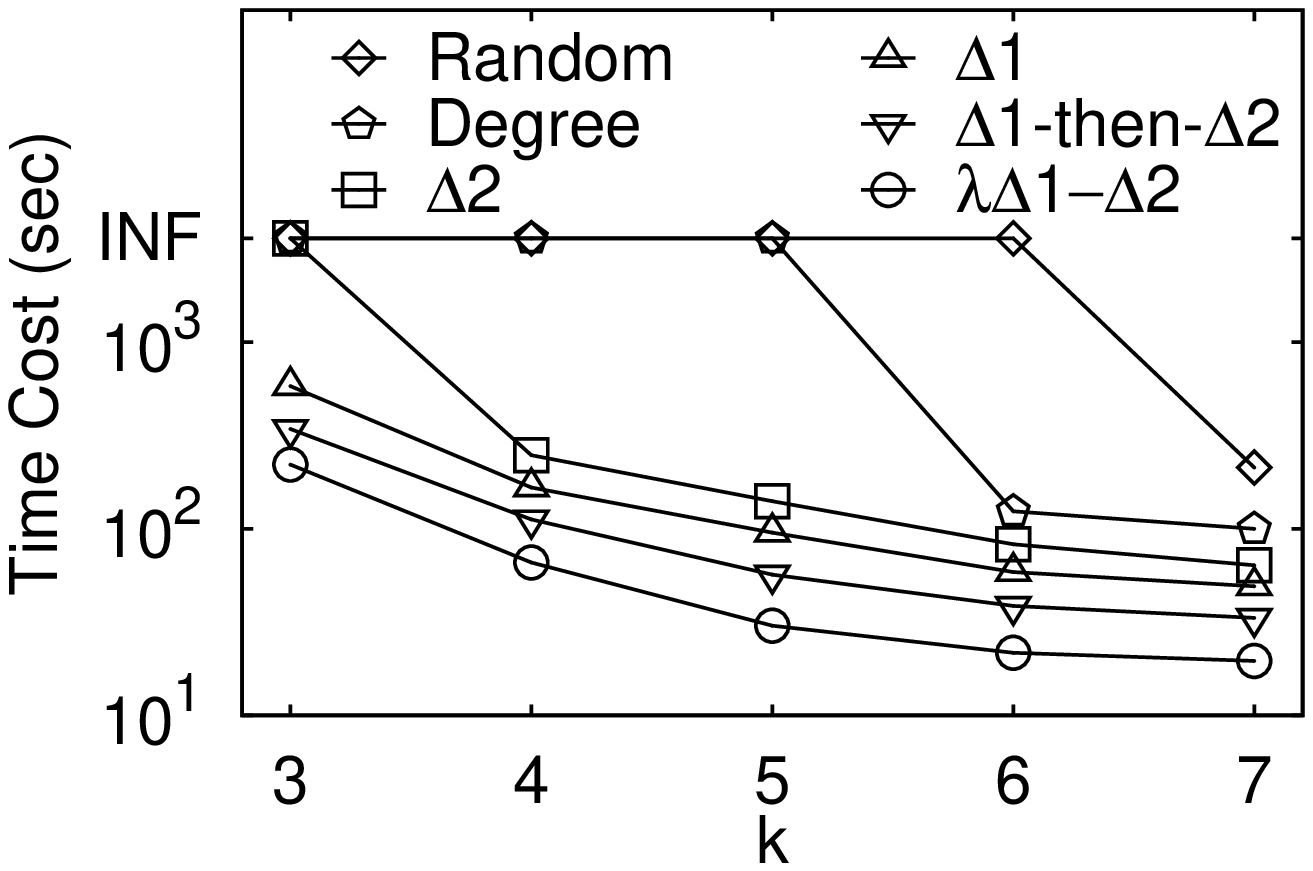}}
    \subfigure[Enumeration, Gow., k=5]{
    \includegraphics[width=0.48\columnwidth, height=2.5cm]{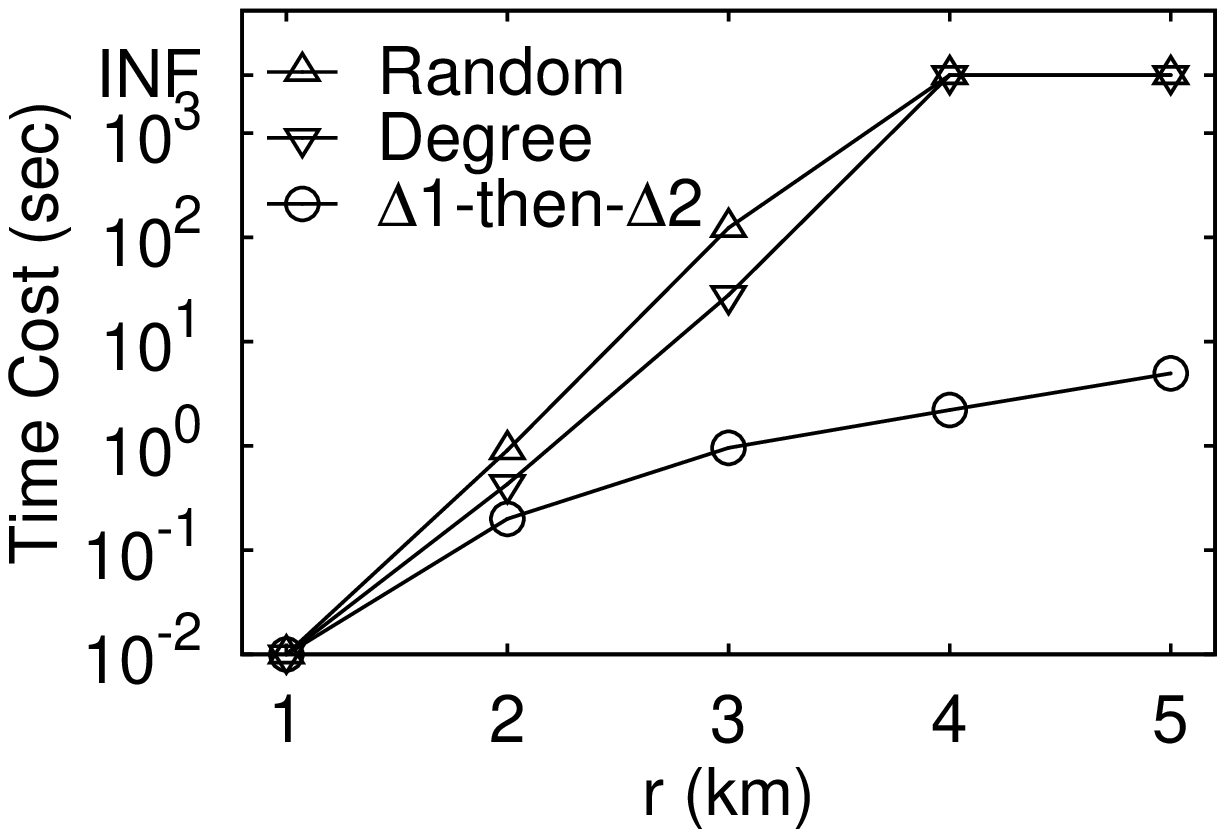}}
    \vspace{-2mm}

    \subfigure[Enumeration, Gow., k=5]{
    \includegraphics[width=0.48\columnwidth, height=2.5cm]{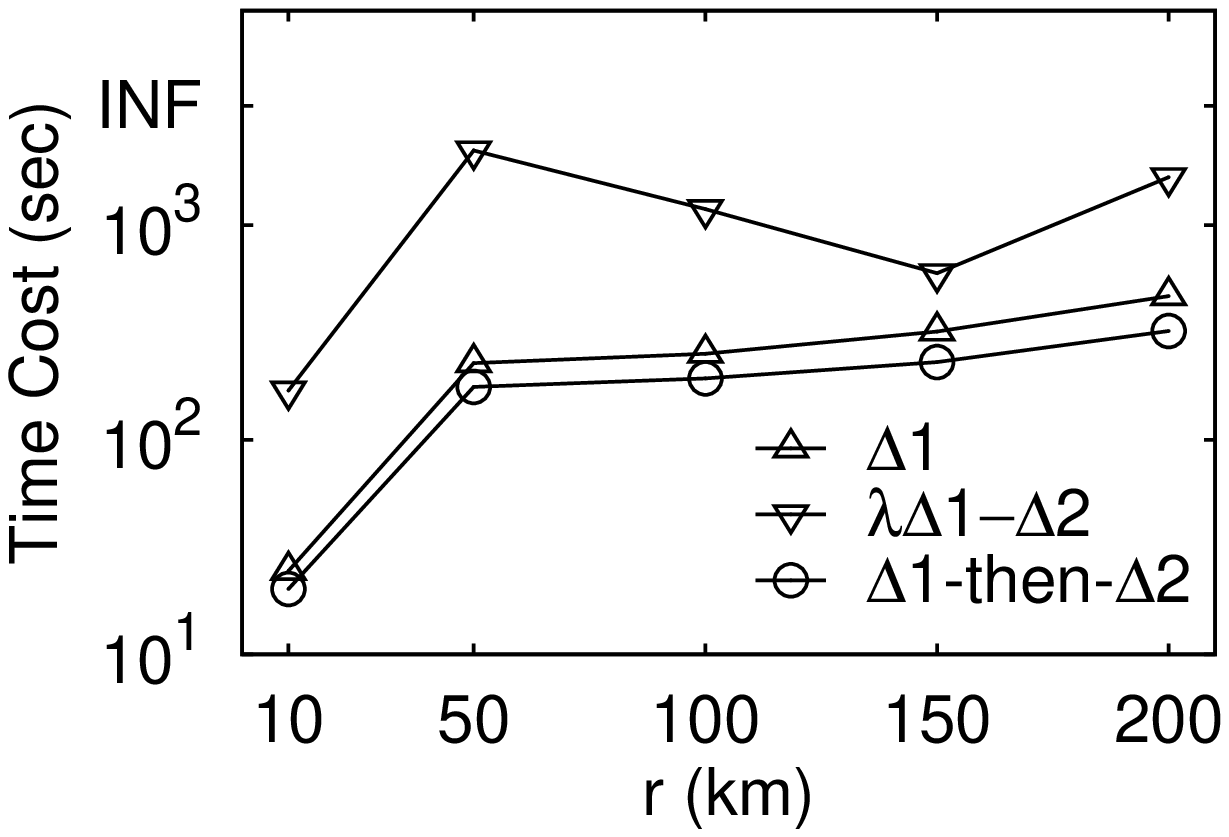}}
    \subfigure[Maximal, Gow., k=5]{
    \includegraphics[width=0.48\columnwidth, height=2.5cm]{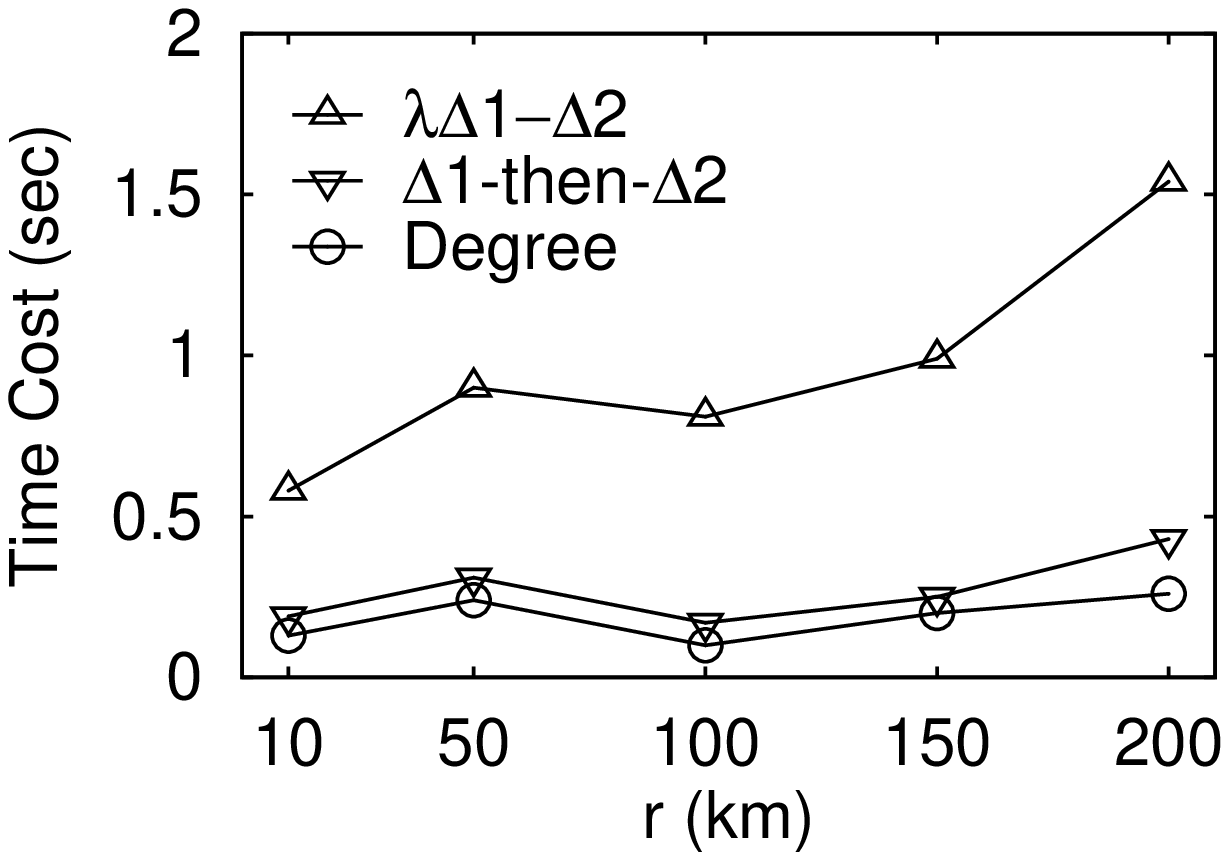}}
\end{center}
\vspace{-7mm}
\caption{\small{Evaluate Search Orders}}
\label{fig:exp:orders}
\end{figure}

\vspace{1mm}
\noindent \textbf{Evaluating the Search Orders.}
In this experiment, we evaluate the effectiveness of the three search orders
proposed for the maximum algorithm (Section~\ref{subsec:maxorder}, Figure~\ref{fig:exp:orders}(a)-(c)),
enumeration algorithm (Section~\ref{subsec:enuorder}, Figure~\ref{fig:exp:orders}(d)-(e))
and the checking maximal algorithm (Section~\ref{subsec:maxcheckorder}, Figure~\ref{fig:exp:orders}(f)).
We first tune $\lambda$ value for the search order of \advmax in Figure~\ref{fig:exp:orders}(a)
against \texttt{DBLP} and \texttt{Gowalla}. In the following experiments, we set $\lambda$ to $5$ for maximum algorithms.
Figure~\ref{fig:exp:orders}(b) verifies the importance of the adaptive order for the two branches on \texttt{DBLP}
where \textbf{Expand} (resp. \textbf{Shrink}) means the expand (resp. shrink) branch is always preferred in \advmax.
In Figure~\ref{fig:exp:orders}(c), we investigate a set of possible order strategies for \advmax.
As expected, the $\lambda\Delta_1-\Delta_2$ order proposed in Section~\ref{subsec:maxorder} outperforms
the other alternatives including random order, degree based order (Section~\ref{subsec:maxcheckorder}, used for checking maximals), $\Delta_1$ order, $\Delta_2$ order and $\Delta_1$-then-$\Delta_2$ order (Section~\ref{subsec:enuorder}, used by \advenum).
Similarly, Figure~\ref{fig:exp:orders}(d) and (e) confirm that the $\Delta_1$-then-$\Delta_2$ order
is the best choice for \advenum compared to the alternatives.
Figure~\ref{fig:exp:orders}(f) shows that the degree order achieves the best performance for the checking maximal algorithm
(Algorithm~\ref{alg:checkMaximal}) compared to the two orders used by \advenum and \advmax.

\begin{figure}[htb]
\begin{center}
     \subfigure[Enumeration]{
     \includegraphics[width=0.48\columnwidth, height=2.5cm]{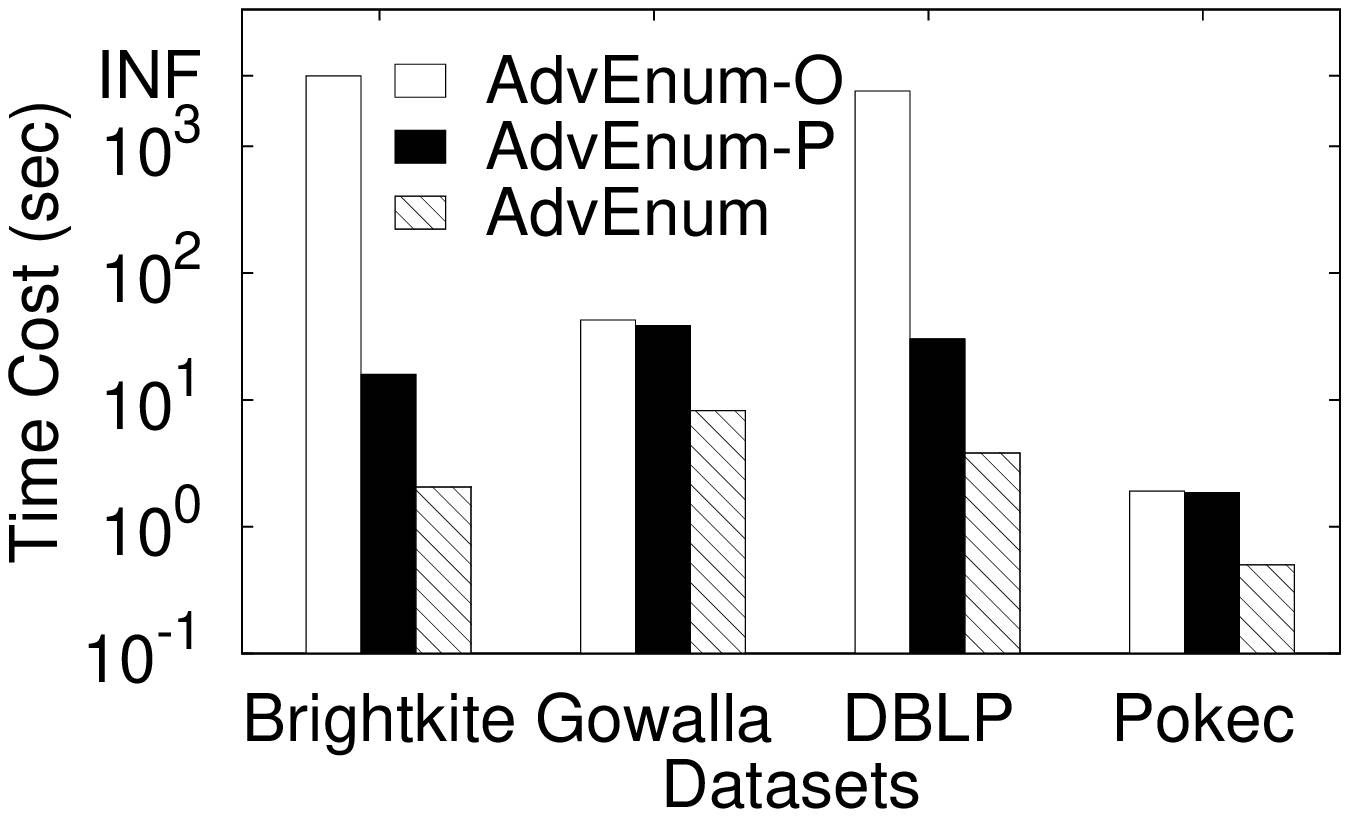}}
     \subfigure[Maximum]{
     \includegraphics[width=0.48\columnwidth, height=2.5cm]{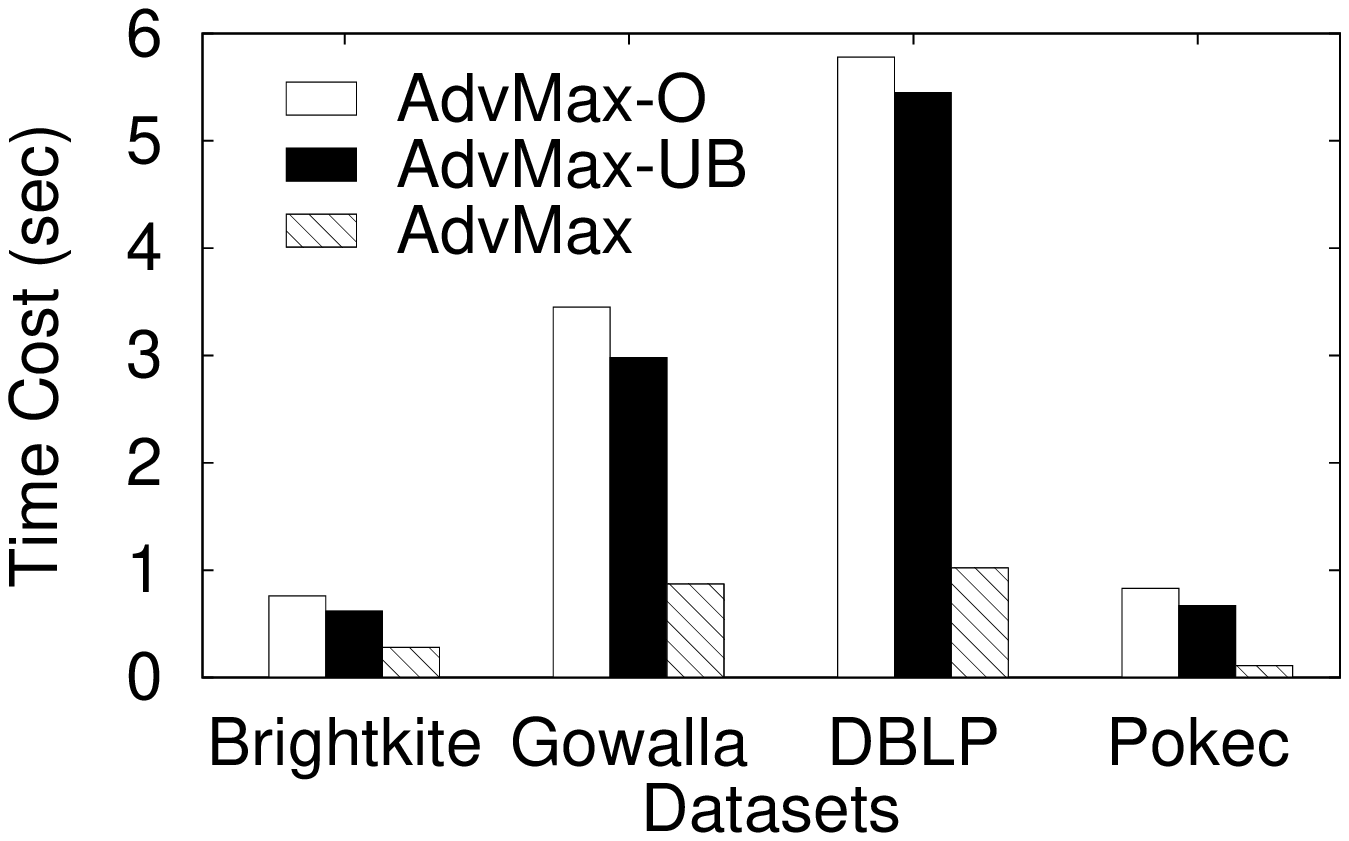}}
\end{center}
\vspace{-7mm}
\caption{\small{Performance on Four Datasets}}
\label{fig:exp:datasetsPerformance}
\end{figure}

\vspace{1mm}
\noindent \textbf{Effect of Different Datasets.}
Figure~\ref{fig:exp:datasetsPerformance} evaluates the performance of the enumeration and maximum algorithms on four datasets with $k=10$.
We set $r$ to $500$ km, $300$ km, $3$\textperthousand~and $5$\textperthousand~in \texttt{Brightkite}, \texttt{Gowalla}, \texttt{DBLP} and \texttt{Pokec}, respectively.
We use \textbf{AdvEnum-O} to denote the \advenum algorithm \textit{without} the best order
while all other advanced techniques applied (degree order is used instead).
Similarly, \textbf{AdvEnum-P} stands for the \advenum algorithm without the support of the advanced techniques
(candidate retention, early termination and checking maximals), but the best search order is employed.
Figure~\ref{fig:exp:datasetsPerformance}(a) demonstrates the efficiency of those advanced techniques and search orders on four datasets.
We also demonstrate the efficiency of the upper bound and search order for the maximum algorithm in Figure~\ref{fig:exp:datasetsPerformance}(b).
where three algorithms are evaluated - \textbf{AdvMax-O}, \textbf{AdvMax-UB}, and \advmax.
\begin{figure}[thb]
\begin{center}
	\begin{minipage}[b]{\columnwidth}
		\centering
		 \includegraphics[width=1.0\columnwidth]{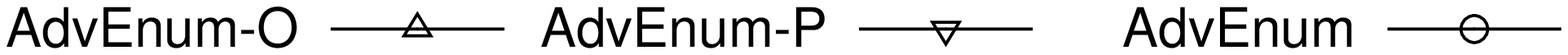}%
	\end{minipage}\vspace{-2mm}

    \subfigure[Gowalla, r=100]{
    \includegraphics[width=0.48\columnwidth, height=2.5cm]{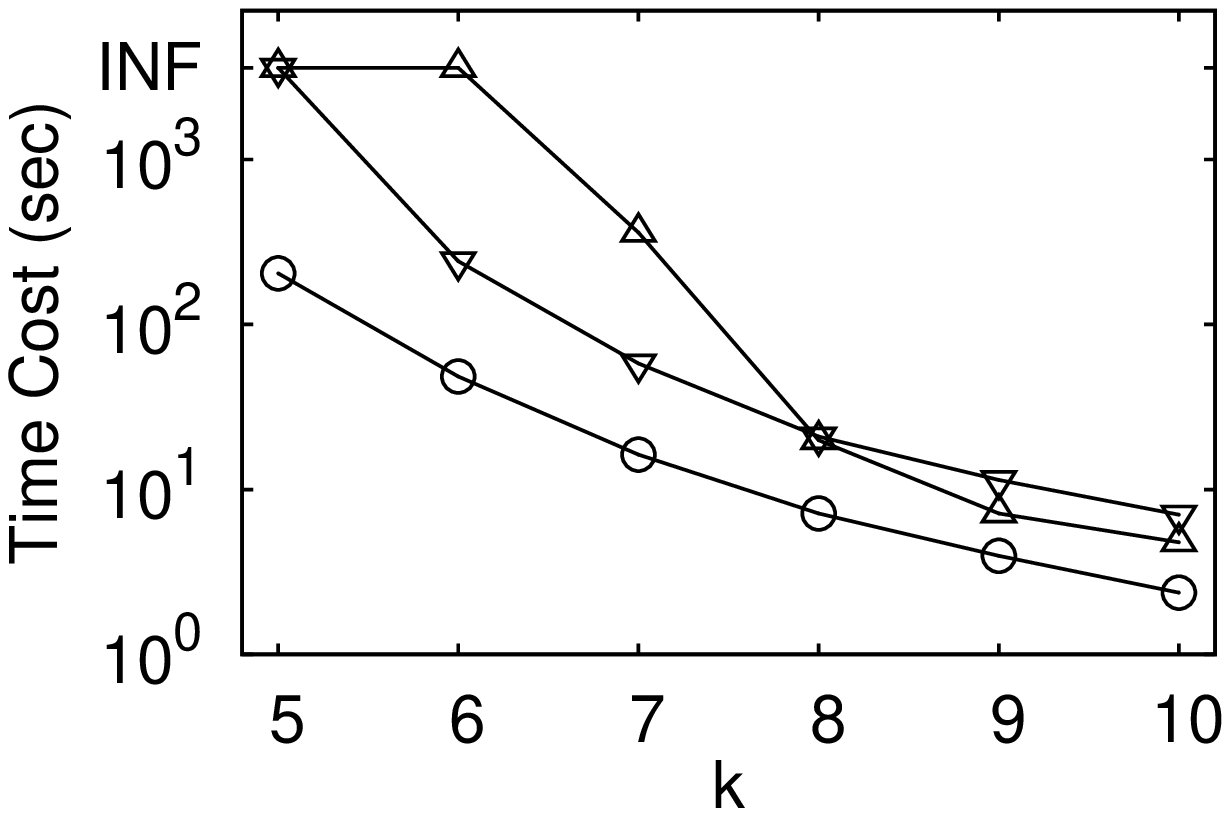}}
    \subfigure[DBLP, k=15]{
    \includegraphics[width=0.48\columnwidth, height=2.5cm]{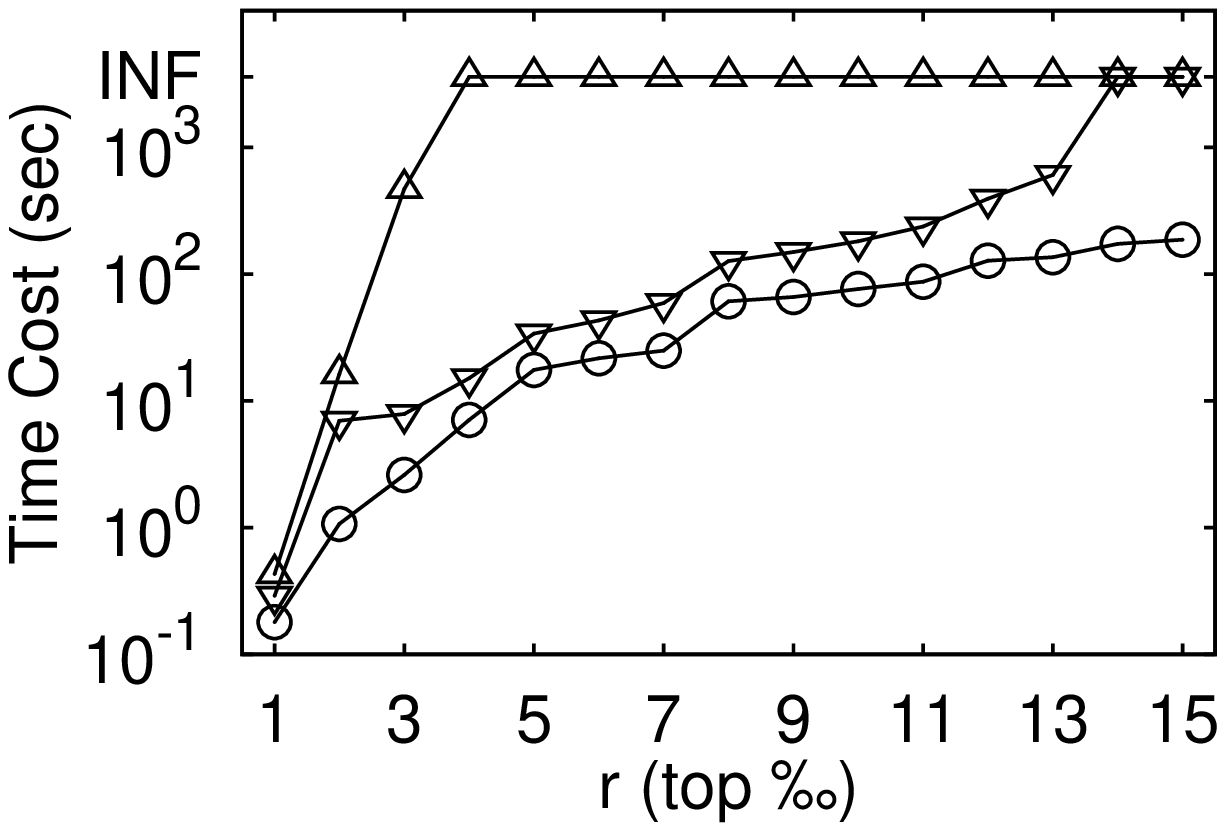}}
\end{center}
\vspace{-7mm}
\caption{\small{Effect of $k$ and $r$ for Enumeration}}
\label{fig:exp:enuPerformance}
\end{figure}
\begin{figure}[thb]
\begin{center}
	\begin{minipage}[b]{\columnwidth}
		\centering
		 \includegraphics[width=1.0\columnwidth]{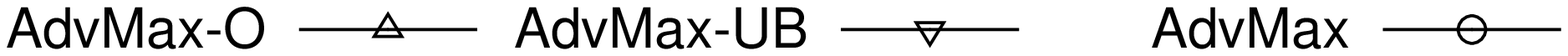}%
	\end{minipage}\vspace{-2mm}

    \subfigure[Gowalla, r=100]{
    \includegraphics[width=0.48\columnwidth, height=2.5cm]{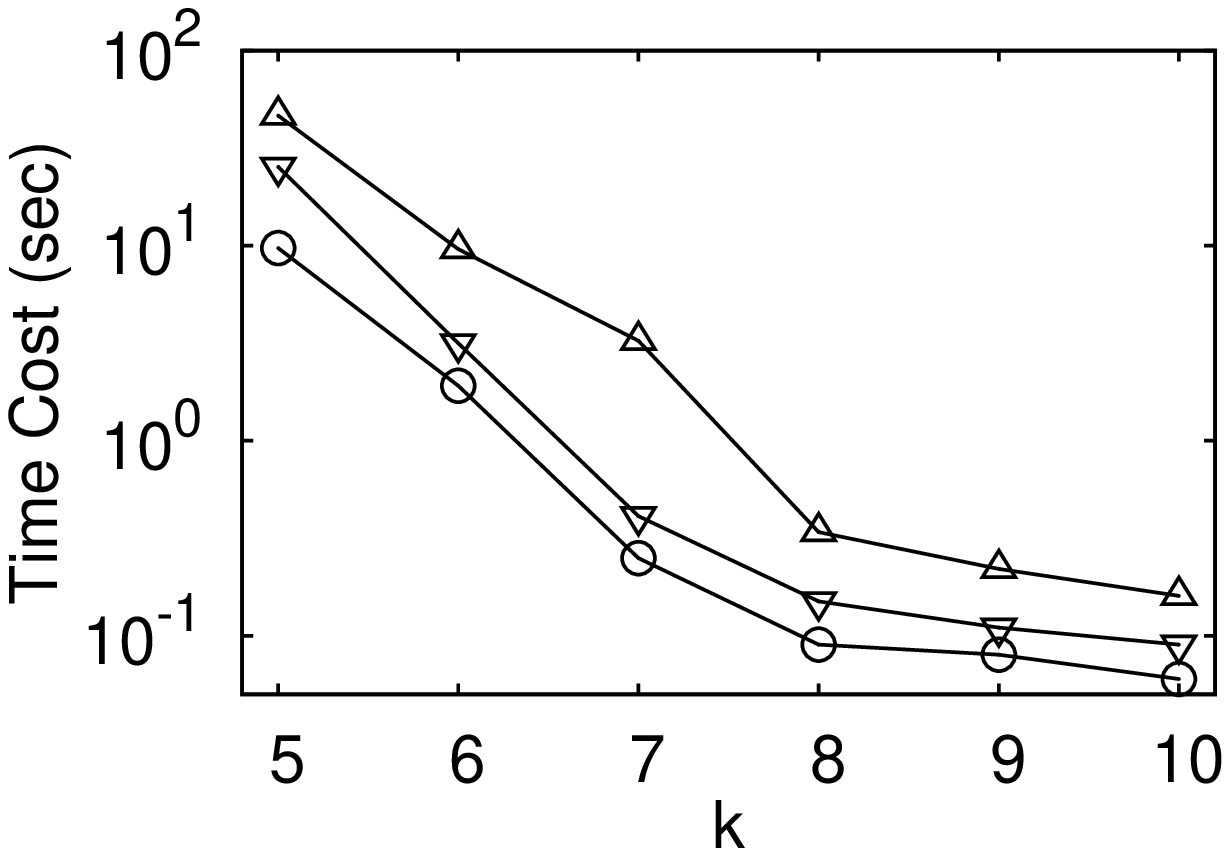}}
    \subfigure[DBLP, k=15]{
    \includegraphics[width=0.48\columnwidth, height=2.5cm]{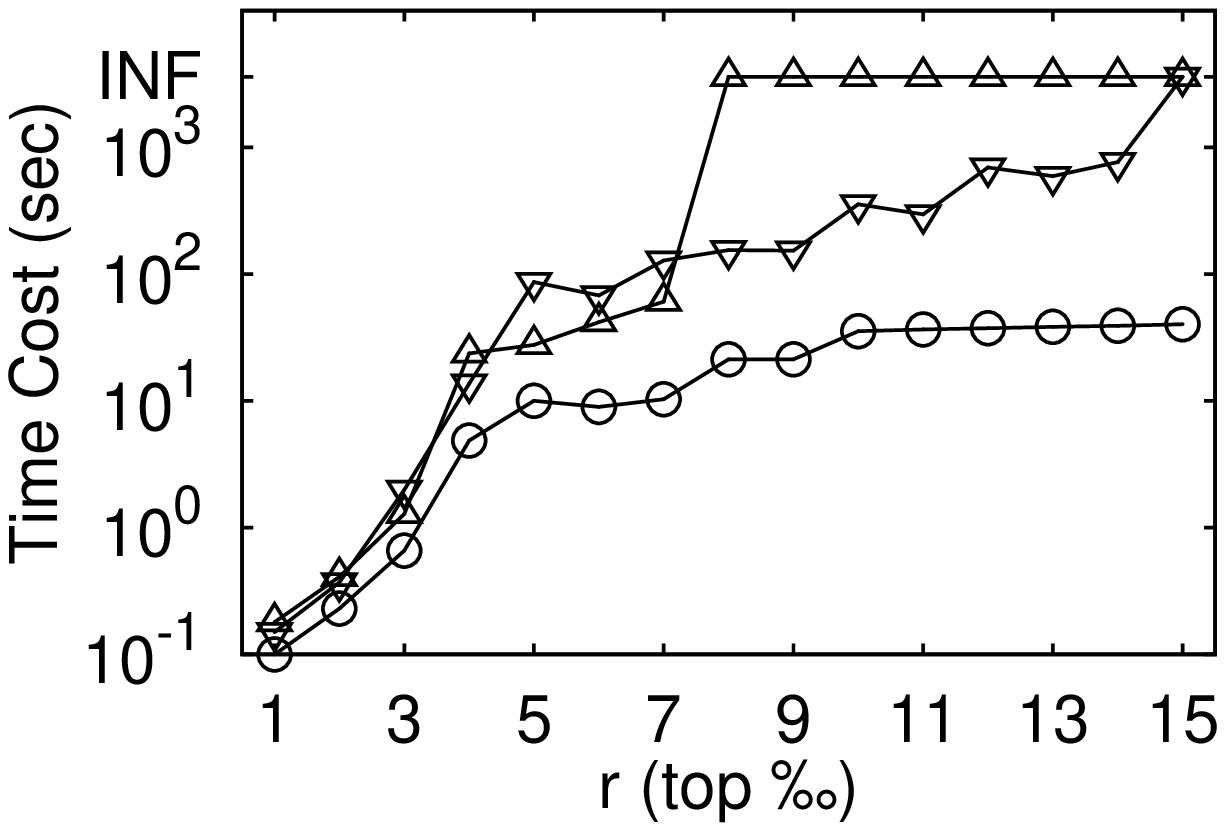}}
\end{center}
\vspace{-7mm}
\caption{\small{Effect of $k$ and $r$ for Maximum}}
\label{fig:exp:maxPerformance}
\end{figure}

\vspace{1mm}
\noindent \textbf{Effect of $k$ and $r$.}
Figure~\ref{fig:exp:enuPerformance} studies the impact of $k$ and $r$ for the three enumeration algorithms on \texttt{Gowalla} and \texttt{DBLP}.
As expected, Figure~\ref{fig:exp:enuPerformance}(a) shows that the time cost of three algorithms drops when $k$ grows
because many more vertices are pruned by the structure constraint.
In Figure~\ref{fig:exp:enuPerformance}(b), the time costs of three algorithms grow when $r$ increases
because more vertices will contribute to the same \krc when the similarity threshold drops.
Similar trends are also observed in Figure~\ref{fig:exp:maxPerformance} where three maximum algorithms with different $k$ and $r$ values are evaluated.
Moreover, Figure~\ref{fig:exp:enuPerformance} and~\ref{fig:exp:maxPerformance}
further confirm the efficiency of our pruning techniques, the \kkc based upper bound technique,
and our search orders.
As expected, we also observed that \advmax significantly outperforms \advenum under the same setting
because \advmax can further cut-off search trees based on the derived upper bound of the core size
and does not need to check the maximals.

%
%

\section{Related Work}
\label{sec:rel}
%
The first work of \kc was studied in~\cite{seidman1983network} and the model has been widely used in many applications such as social contagion~\cite{ugander2012structural}, graph visualization~\cite{DBLP:journals/pvldb/ZhaoT12}, influence study~\cite{kitsak2010identification}, and user engagement~\cite{DBLP:journals/siamdm/BhawalkarKLRS15,DBLP:conf/cikm/MalliarosV13}.
Batagelj and Zaversnik presented a liner in-memory algorithm for core decomposition~\cite{DBLP:journals/corr/cs-DS-0310049}.
I/O efficient algorithms~\cite{DBLP:conf/icde/ChengKCO11,wen2015efficient} were proposed for core decomposition on graphs that cannot fit in the main memory. Locally computing and estimating core numbers are studied in~\cite{DBLP:conf/sigmod/CuiXWW14} and~\cite{DBLP:conf/icdm/OBrienS14} respectively.
Besides, pairwise similarity has been widely used to identify groups of similar users based on their attributes (e.g.,~\cite{jaho2011iscode,rangapuram2013towards}). 
A variety of clique computation algorithms have been proposed in the literature (e.g.,~\cite{cheng2012fast,wang2013redundancy}).


Due to the popularity of attributed graphs in various applications,
a large amount of classical graph queries have been investigated in this context
such as clustering~\cite{DBLP:conf/sdm/AkogluTMF12,DBLP:conf/sigmod/XuKWCC12},
community detection~\cite{dang2012community,DBLP:conf/icdm/YangML13},
pattern mining and matching~\cite{DBLP:journals/pvldb/SilvaMZ12,DBLP:conf/kdd/TongFGE07},
event detection~\cite{DBLP:conf/kdd/RozenshteinAGT14},
and network modeling~\cite{DBLP:conf/www/PfeifferMFNG14}.
Nevertheless, the problem studied in this paper is inherently different to these works,
as none of the them consider $k$-core computation on attributed graphs.

Recently, Wu {\emph et al.}~\cite{DBLP:conf/icde/WuJZZ15} developed efficient algorithms to find dense and connected subgraphs in dual networks. Nevertheless, their model is inherently different to our \krc model because
they only consider the cohesiveness of the dual graph (i.e., the attribute similarity in our problem)
based on densest subgraph model.
The connectivity constraint on graph structure alone cannot reflect the structure cohesiveness of the graph.
Recently, Zhu {\emph et al.}~\cite{DBLP:journals/corr/ZhuHXL14} studied $k$-core computation within a given spatial region in the context of geo-social networks.
However, we consider the similarity of the attributes instead of a specific regions.
Lee {\emph et al.}~\cite{lee2014cast} proposed a model based on \kc called ($k$,$d$)-core which is different with our model because they consider structure and similarity constraint on same edges and use the number of common neighbors of two vertices in each edge as the similarity constraint.
A recent community search paper~\cite{DBLP:journals/pvldb/FangCLH16} finds a subgraph related to a query point considering cohesiveness on both structure and keyword similarity, while it focuses on maximizing the number of common keywords of the vertices in the subgraph. The techniques are different with ours and cannot be applied to solving our problem.

\section{Conclusion}
\label{sec:con}
In this paper, we propose a novel cohesive subgraph model, called \krc,
which considers the cohesiveness of a subgraph from the perspective of both graph structure and vertex attribute.
We show that the problem of enumerating the maximal \krcs and finding the maximum \krc are both NP-hard.
Several novel pruning techniques are proposed to improve algorithm efficiency,
including candidate size reduction, early termination, checking maximals and upper bound estimation techniques.
We also devise effective search orders for enumeration, maximum and maximal check algorithms.
Extensive experiments on real-life networks demonstrate the effectiveness of the \krc model, as well
as the efficiency of our techniques.

{\small
\bibliographystyle{abbrv}
\bibliography{ref}

\begin{thebibliography}{10}

\bibitem{DBLP:conf/sdm/AkogluTMF12}
L.~Akoglu, H.~Tong, B.~Meeder, and C.~Faloutsos.
\newblock {PICS:} parameter-free identification of cohesive subgroups in large
  attributed graphs.
\newblock In {\em ICDM}, pages 439--450, 2012.

\bibitem{DBLP:journals/corr/cs-DS-0310049}
V.~Batagelj and M.~Zaversnik.
\newblock An o(m) algorithm for cores decomposition of networks.
\newblock {\em CoRR}, cs.DS/0310049, 2003.

\bibitem{DBLP:journals/siamdm/BhawalkarKLRS15}
K.~Bhawalkar, J.~M. Kleinberg, K.~Lewi, T.~Roughgarden, and A.~Sharma.
\newblock Preventing unraveling in social networks: The anchored k-core
  problem.
\newblock {\em {SIAM} J. Discrete Math.}, 29(3):1452--1475, 2015.

\bibitem{DBLP:conf/icde/ChengKCO11}
J.~Cheng, Y.~Ke, S.~Chu, and M.~T. {\"{O}}zsu.
\newblock Efficient core decomposition in massive networks.
\newblock In {\em ICDE}, pages 51--62, 2011.

\bibitem{cheng2012fast}
J.~Cheng, L.~Zhu, Y.~Ke, and S.~Chu.
\newblock Fast algorithms for maximal clique enumeration with limited memory.
\newblock In {\em SIGKDD}, pages 1240--1248, 2012.

\bibitem{DBLP:journals/iandc/ChitnisFG16}
R.~Chitnis, F.~V. Fomin, and P.~A. Golovach.
\newblock Parameterized complexity of the anchored k-core problem for directed
  graphs.
\newblock {\em Inf. Comput.}, 247:11--22, 2016.

\bibitem{DBLP:conf/aaai/ChitnisFG13}
R.~H. Chitnis, F.~V. Fomin, and P.~A. Golovach.
\newblock Preventing unraveling in social networks gets harder.
\newblock In {\em AAAI}, 2013.

\bibitem{DBLP:conf/sigmod/CuiXWW14}
W.~Cui, Y.~Xiao, H.~Wang, and W.~Wang.
\newblock Local search of communities in large graphs.
\newblock In {\em SIGMOD}, pages 991--1002, 2014.

\bibitem{dang2012community}
T.~Dang and E.~Viennet.
\newblock Community detection based on structural and attribute similarities.
\newblock In {\em International Conference on Digital Society (ICDS)}, pages
  7--12, 2012.

\bibitem{DBLP:journals/pvldb/FangCLH16}
Y.~Fang, R.~Cheng, S.~Luo, and J.~Hu.
\newblock Effective community search for large attributed graphs.
\newblock {\em {PVLDB}}, 9(12):1233--1244, 2016.

\bibitem{garey1976complexity}
M.~R. Garey and D.~S. Johnson.
\newblock The complexity of near-optimal graph coloring.
\newblock {\em JACM}, 23(1):43--49, 1976.

\bibitem{DBLP:books/fm/GareyJ79}
M.~R. Garey and D.~S. Johnson.
\newblock {\em Computers and Intractability: {A} Guide to the Theory of
  NP-Completeness}.
\newblock W. H. Freeman, 1979.

\bibitem{DBLP:conf/www/PfeifferMFNG14}
J.~J.~P. III, S.~Moreno, T.~L. Fond, J.~Neville, and B.~Gallagher.
\newblock Attributed graph models: modeling network structure with correlated
  attributes.
\newblock In {\em WWW}, pages 831--842, 2014.

\bibitem{jaho2011iscode}
E.~Jaho, M.~Karaliopoulos, and I.~Stavrakakis.
\newblock Iscode: a framework for interest similarity-based community detection
  in social networks.
\newblock In {\em INFOCOM WKSHPS}, pages 912--917, 2011.

\bibitem{kitsak2010identification}
M.~Kitsak, L.~K. Gallos, S.~Havlin, F.~Liljeros, L.~Muchnik, H.~E. Stanley, and
  H.~A. Makse.
\newblock Identification of influential spreaders in complex networks.
\newblock {\em Nature physics}, 6(11):888--893, 2010.

\bibitem{lee2014cast}
P.~Lee, L.~V.~S. Lakshmanan, and E.~E. Milios.
\newblock {CAST:} {A} context-aware story-teller for streaming social content.
\newblock In {\em {CIKM}}, pages 789--798, 2014.

\bibitem{DBLP:conf/cikm/MalliarosV13}
F.~D. Malliaros and M.~Vazirgiannis.
\newblock To stay or not to stay: modeling engagement dynamics in social
  graphs.
\newblock In {\em {CIKM}}, pages 469--478, 2013.

\bibitem{DBLP:conf/icdm/OBrienS14}
M.~P. O'Brien and B.~D. Sullivan.
\newblock Locally estimating core numbers.
\newblock In {\em ICDM}, pages 460--469, 2014.

\bibitem{rangapuram2013towards}
S.~S. Rangapuram, T.~B{\"u}hler, and M.~Hein.
\newblock Towards realistic team formation in social networks based on densest
  subgraphs.
\newblock In {\em WWW}, pages 1077--1088, 2013.

\bibitem{DBLP:conf/kdd/RozenshteinAGT14}
P.~Rozenshtein, A.~Anagnostopoulos, A.~Gionis, and N.~Tatti.
\newblock Event detection in activity networks.
\newblock In {\em SIGKDD}, pages 1176--1185, 2014.

\bibitem{seidman1983network}
S.~B. Seidman.
\newblock Network structure and minimum degree.
\newblock {\em Social networks}, 5(3):269--287, 1983.

\bibitem{DBLP:journals/pvldb/SilvaMZ12}
A.~Silva, W.~M. Jr., and M.~J. Zaki.
\newblock Mining attribute-structure correlated patterns in large attributed
  graphs.
\newblock {\em {PVLDB}}, 5(5), 2012.

\bibitem{DBLP:conf/kdd/TongFGE07}
H.~Tong, C.~Faloutsos, B.~Gallagher, and T.~Eliassi{-}Rad.
\newblock Fast best-effort pattern matching in large attributed graphs.
\newblock In {\em {SIGKDD}}, pages 737--746, 2007.

\bibitem{ugander2012structural}
J.~Ugander, L.~Backstrom, C.~Marlow, and J.~Kleinberg.
\newblock Structural diversity in social contagion.
\newblock {\em PNAS}, 109(16):5962--5966, 2012.

\bibitem{wang2013redundancy}
J.~Wang, J.~Cheng, and A.~W. Fu.
\newblock Redundancy-aware maximal cliques.
\newblock In {\em {SIGKDD}}, pages 122--130, 2013.

\bibitem{wen2015efficient}
D.~Wen, L.~Qin, Y.~Zhang, X.~Lin, and J.~X. Yu.
\newblock {I/O} efficient core graph decomposition at web scale.
\newblock In {\em {ICDE}}, pages 133--144, 2016.

\bibitem{DBLP:conf/wsdm/WuSFLT13}
S.~Wu, A.~D. Sarma, A.~Fabrikant, S.~Lattanzi, and A.~Tomkins.
\newblock Arrival and departure dynamics in social networks.
\newblock In {\em WSDM}, pages 233--242, 2013.

\bibitem{DBLP:conf/icde/WuJZZ15}
Y.~Wu, R.~Jin, X.~Zhu, and X.~Zhang.
\newblock Finding dense and connected subgraphs in dual networks.
\newblock In {\em {ICDE}}, pages 915--926, 2015.

\bibitem{DBLP:conf/sigmod/XuKWCC12}
Z.~Xu, Y.~Ke, Y.~Wang, H.~Cheng, and J.~Cheng.
\newblock A model-based approach to attributed graph clustering.
\newblock In {\em SIGMOD}, pages 505--516, 2012.

\bibitem{DBLP:conf/icdm/YangML13}
J.~Yang, J.~J. McAuley, and J.~Leskovec.
\newblock Community detection in networks with node attributes.
\newblock In {\em ICDM}, pages 1151--1156, 2013.

\bibitem{DBLP:conf/icde/YuanQLCZ15}
L.~Yuan, L.~Qin, X.~Lin, L.~Chang, and W.~Zhang.
\newblock Diversified top-k clique search.
\newblock In {\em {ICDE}}, pages 387--398, 2015.

\bibitem{DBLP:journals/pvldb/ZhaoT12}
F.~Zhao and A.~K.~H. Tung.
\newblock Large scale cohesive subgraphs discovery for social network visual
  analysis.
\newblock {\em PVLDB}, 6(2), 2012.

\bibitem{DBLP:journals/corr/ZhuHXL14}
Q.~Zhu, H.~Hu, J.~Xu, and W.~Lee.
\newblock Geo-social group queries with minimum acquaintance constraint.
\newblock {\em CoRR}, abs/1406.7367, 2014.

\end{thebibliography}
}
\end{document}